\documentclass[11pt, oneside]{article}   	
\usepackage{geometry,amsmath,amssymb,amsthm}                		
\usepackage{tikz}

\usepackage{apxproof}

\newcommand{\longv}[1]{}

\newcommand{\loglog}{\ell_2}
\newcommand{\ACDL}{\mathbb{ACDL}}

\title{Characterizing Small Circuit Classes from $\AC^0$ to $\AC^1$ via Discrete Ordinary Differential Equations} 

\author{Melissa Antonelli$^{1,2}$ \quad Arnaud Durand$^3$ \quad Juha Kontinen$^2$}

\date{\small{$^1$Helsinki Institute for Information Technology, $^2$University of Helsinki, $^3$Universit\'e Paris Cit\'e}}

\usepackage{amsthm,xcolor,hyperref}

\newcommand{\BIT}{\mathsf{BIT}}

\newcommand{\fun}[1]{\mathsf{#1}}

\newcommand{\sB}{B}
\newcommand{\sA}{A}

\newcommand{\Nat}{\mathbb{N}}

\newcommand{\FP}{\mathbf{FP}}
\newcommand{\NC}{\mathbf{FNC}}
\newcommand{\AC}{\mathbf{FAC}}
\newcommand{\TC}{\mathbf{FTC}}
\newcommand{\FACC}{\mathbf{FACC[2]}}

\newcommand{\sign}{\mathsf{sg}}

\newcommand{\tu}[1]{\mathbf{#1}}

\newtheorem{prop}{Proposition}
\newtheorem{theorem}{Theorem}
\newtheorem{lemma}{Lemma}

\theoremstyle{definition}
\newtheorem{defn}{Definition}
\newtheorem{notation}{Notation}
\newtheorem{example}{Example}
\newtheorem{remark}{Remark}

\newcommand{\ODEACC}{\ell\text{-ODE}_{\text{FACC[2]}}}

\newcommand{\schODE}{\ell\text{-b}_0\text{ODE}}

\newcommand{\chODE}{\ell\text{-bODE}}

\newcommand{\pODE}{\ell\text{-pODE}}

\begin{document}

\maketitle

\begin{abstract}
    In this paper, we provide a uniform framework for investigating small circuit classes and bounds through the lens of ordinary differential equations (ODEs).
    Following an approach recently introduced to capture the class of polynomial-time computable functions via ODE-based recursion schemas and later applied to the context of functions computed by unbounded fan-in circuits of constant depth ($\AC^0$), 
     we study multiple relevant small circuit classes. 
    In particular, we show that \emph{natural} restrictions on linearity and derivation along functions with specific growth rate correspond to kinds of functions that can be proved to be  in various classes, ranging from $\AC^0$ to $\AC^1$.
    This reveals an intriguing link between constraints over linear-length ODEs and circuit computation, providing new tools to tackle the complex challenge of establishing bounds for classes in the circuit hierarchies and possibly enhancing our understanding of the role of counters in this setting.
    Additionally, we establish several completeness results, in particular obtaining the first ODE-based characterizations for the classes of functions computable in constant depth with unbounded fan-in and \textsc{Mod 2} gates ($\FACC$) and in logarithmic depth with bounded fan-in Boolean gates ($\NC^1$).
\end{abstract}


\section{Introduction}

Implicit computational complexity is a active area of theoretical computer science that aims to provide machine-independent characterizations of relevant complexity classes.
In the context of recursion theory, a foundational work in this field is by Cobham~\cite{Cobham}, who, in the 1960s, provided an implicit characterization of the class of polynomial-time computable functions ($\FP$) through a \emph{limited} recursion schema.
This work not only led to the development of function algebras characterizing multiple classes different from poly-time (see~\cite{CloteKranakis}), but also inspired several alternative recursion-theoretic approaches~\cite{BellantoniCook,Leivant,LeivantMarion93}.
Among them, the one introduced in~\cite{BournezDurand19}, shows that ordinary differential equations (ODEs) provide a natural tool for algorithmic design and offers an original characterization of $\FP$, using \emph{linear} systems of equations and \emph{derivation along a logarithmically growing function} (such as the length function $\ell$).
Informally, the latter condition controls the number of recursion steps, while linearity governs the growth of objects generated during the computation.
This approach has recently been generalized to the continuous setting~\cite{BlancBournez23,BlancBournez24} and to the study of two small circuit classes~\cite{ADK24a}.
In this paper, we extend the latter far further by presenting a \emph{comprehensive} investigation of small circuit classes through the lens of ODE schemas.

Although small circuit classes have been characterized in different ways, investigating them through the ODE framework has never been attempted (except for~\cite{ADK24a}).
This approach is both \emph{interesting}, as for a descriptive method based on ODEs to be meaningful and fruitful, it must be able to adjust with very subtle changes in restricted models of computation, and \emph{challenging}, because even simple and useful mathematical functions may not be computable within the considered classes and, consequently, the tools available may seem drastically restricted and the naturalness of the approach questionable.
However, the results of this paper show that this is not hopeless. 
Very simple and natural restrictions on linearity – namely, constraining the values of coefficients and allowing (or not) the call for the sign of the defined function inside the linear term – and on the function along which derivation is performed, allows us to capture the very essence of small circuit computation.
%
%
Many classical functions – such as polylog bit sum, parity, majority, iterated sum and product, and binary search – are shown to be naturally encapsulated in very simple linear ODE families.   
For this reason, we believe that our approach to this line of research, which is still in its beginning, can shed some new light on circuit computation and offer a simple (and machine-independent) framework for analyzing these notoriously hard problems through very \emph{natural} and \emph{uniform} constraints on linear ODEs.

Concretely, the results of this paper have led to completely original ODE-based characterizations for the classes of functions computable by unbounded fan-in circuits of constant depth and polynomial size that allow \textsc{Mod 2} gates ($\FACC$) and by bounded fan-in circuits of logarithmic depth and polynomial size ($\NC^1$).
We have also improved our understanding of schemas capturing the classes of functions computable by unbounded fan-in circuits of constant depth and polynomial size ($\AC^0$), possibly including \textsc{Maj} gates ($\TC^0$), by generalizing and making them part of a more general picture.
Finally, we have shown that considering schemas with \emph{derivative along $\ell_2=\ell\circ\ell$} is definitely worth investigating: in particular, the corresponding linear schema is in $\NC^1$, somewhat mirroring what $\ell$-ODE does for $\FP$.

\textbf{Structure of the paper.}
In Section~\ref{sec:preliminaries}, we present the basics of the method developed in~\cite{BournezDurand19}, the necessary notions to formulate our results, and recap definitions of the small circuit classes we will consider throughout the paper.
In Section~\ref{sec:towards}, we present, in a high-level yet systematic and uniform manner, the main schemas (differently) corresponding to the classes that will be studied.
In Section~\ref{sec:alongL}, we introduce ODE-based schemas obtained by deriving along $\ell$ and imposing different constraints on linearity, and we show that they are computable in the respective circuit classes.
We also provide improved completeness results for $\AC^0$ and $\TC^0$, along with completely original ODE-based characterizations for $\FACC$ and $\NC^1$.
In Section~\ref{sec:2ell}, we introduce a new class of ODE schemas, the defining feature of which consists of deriving along $\ell_2$.
Our main result here is introduction of a very natural linear $\ell_2$-ODE the computation of which can be performed in $\NC^1$.
We conclude in Section~\ref{sec:conclusion} pointing to possible directions for future research.
For space reasons, most proofs are omitted or only sketched in the paper,  but full details can be found in the Appendix. 
To make the exposition shorter, most of our characterizations are proved indirectly, 
as these proofs tend to be more compact.
However, direct proofs (both in the uniform and non-uniform settings) can be provided (but are not presented here).

\section{Preliminaries}\label{sec:preliminaries}

We assume the reader familiar with the basics of complexity theory.
In this section, we briefly introduce the approach to complexity delineated in~\cite{BournezDurand23} (see the original paper for full details) and a few notions needed for our new characterizations.
We also cursorily recall the definitions of the circuit classes that we will consider in the remainder of the paper.

\subsection{Capturing Complexity via ODEs}

One of the main approaches developed in implicit computational complexity is that based on recursion.
It was initiated by Cobham~\cite{Cobham}, who gave the first characterization for $\FP$ via a limited recursion schema. 
This groundbreaking result has inspired not only similar characterizations for other classes, but also alternative recursion-based approaches to capture poly-time computable functions, such as safe recursion~\cite{BellantoniCook} and ramification~\cite{Leivant,LeivantMarion}.
Among them, the proposal by~\cite{BournezDurand19} offers the advantages of not imposing any explicit bound on the recursion schema or assigning specific role to variables.
Instead, it is based on Discrete Ordinary Differential Equations (a.k.a. Difference Equations) and combines two novel features, namely~derivation along specific functions, to control the number of computation steps, and a special syntactic form of the equation 
(linearity),  to control the object size.

Recall that the \emph{discrete derivative of $\tu f(x)$} is defined as $\Delta \tu f(x) = \tu f(x+1) - \tu f(x)$, and that ODEs are expressions of the form 
$$
\frac{\partial \tu f(x, \tu y)}{\partial x} = \tu h\big(x,\tu y, \tu f(x, \tu y)\big)
$$ 
where $\frac{\partial \tu f(x, \tu y)}{\partial x}$ stands for the derivative of $\tu f(x, \tu y)$ considered as a function of $x$, for a fixed value for $\tu y$.
When some initial value $\tu f(0, \tu y) = \tu g(\tu y)$ is added, this is called Initial Value Problem (IVP).
Then, in order to deal with complexity, some restrictions are needed. 
First, we introduce the notion of \emph{derivation along functions}.

\begin{notation}
    For $x\neq 0$, let $\ell(x)$ denote the length function, returning the length of $x$ written in binary, i.e.~$\lceil$log$_2(x+1)\rceil$, and such that $\ell(0)=0$. 
    We use $\ell_2=\ell\circ\ell$.
    For $u\ge 0$, let $\alpha(u)=2^u-1$ denote the greatest integer whose length is $u$; similarly, $\alpha_2(u)=2^{2^u-1}-1$ is the greatest integer $t$ such that $\ell_2(t)=u$. 
    Finally, let $\div 2$ denote integer division by 2, i.e.~for all $x\in \mathbb{Z}, x \div 2 = \big\lfloor \frac{x}{2}\big\rfloor$.
\end{notation}

\noindent
 Notice that the value of $\ell(x)$ changes (it increases by 1) when $x$ goes from $2^t-1$ to $2^t$ (for $t\ge 1$).
    Similarly, the value of $\ell_2(x)$ changes when $x$ goes from $2^{2^1-1}-1$ to $2^{2^t-1}$ (for $t\ge 1$).

\begin{defn}[$\ell$-ODE Schema]
    Let $\tu f, \lambda :\Nat^p \times \mathbb{Z} \to \mathbb{Z}$ and $\tu h: \Nat^{p+1} \to \mathbb{Z}$ be functions.
    Then,
    \begin{align}
        \frac{\partial \tu f(x, \tu y)}{\partial \lambda} = \frac{\partial \tu f (x, \tu y)}{\partial \lambda (x, \tu y)} = \tu h\big(x, \tu y, \tu f(x, \tu y)\big) 
    \end{align}
    is a formal synonym of $\tu f(x+1, \tu y)=\tu f(x,\tu y)+\big(\lambda (x+1, \tu y) - \lambda (x, \tu y)\big) \times \tu h\big(x, \tu y, \tu f(x, \tu y)\big)$.
    When $\lambda (x, \tu y)=\ell (x)$, we call (1) a \emph{length-ODE} or $\ell$-ODE.
\end{defn}

\noindent
Intuitively, one of the key properties of the $\lambda$-ODE schema is its dependence on the number of distinct values of $\lambda$, i.e.,~the value of $\tu f(x, \tu y)$ changes only when the value of $\lambda(x, \tu y)$ does.
Computation of solutions of $\lambda$-ODEs have been fully described in~\cite{BournezDurand23}.
Here, we focus on the special case of $\lambda=\ell$, which is particularly relevant for our characterizations.
%
If $\tu f$ is a solution of (1) with $\lambda =\ell$ and initial value $\tu f(0, \tu y)=\tu g(\tu y)$, then $\tu f(1, \tu y)=\tu f(0, \tu y) + \tu h\big(\alpha(0), \tu y, \tu f(\alpha(0), \tu y)\big)$.
More generally, for all $x$ and $\tu y$:
$\tu f(x, \tu y) = \tu f(x-1, \tu y) 
+ \Delta \big(\ell(x-1)\big) \times \tu h\big(x-1, \tu y, \tu f(x-1, \tu y)\big)= \tu f\big(\alpha(\ell(x)-1), \tu y\big) + \tu h\big(\alpha(\ell(x)-1), \tu y, \tu f(\alpha(\ell(x)-1), \tu y)\big)$, where $\Delta(\ell(t-1)) =\ell(t)-\ell(t-1)$.
Starting from $t=x\ge 1$ and taking successive decreasing values of $t$, the first difference such that $\Delta (t)\neq 0$ is given by the biggest $t-1$ such that $\ell(t-1) = \ell(x)-1$, i.e.~$t-1=\alpha(\ell(x)-1)$.
Setting $\tu h( \alpha(-1), \tu y, \cdot) = \tu f(0, \tu y)$, it follows by  induction that:
$$
\tu f(x, \tu y) = \sum^{\ell(x)-1}_{u=-1} \tu h\big(\alpha(u), \tu y, \tu f(\alpha(u), \tu y)\big)
$$

%
In the present paper, both polynomial expressions (see~\cite{ADK24a}) and limited polynomial expressions are considered.
Let $\fun{sg}:\mathbb{Z} \to \mathbb{Z}$ be the sign function over $\mathbb{Z}$, taking value 1 for $x>0$ and 0 otherwise, and $\fun{cosg} :\mathbb{Z}\to \mathbb{Z}$ 
be the cosign function, such that $\fun{cosg}(x)=1-\fun{sg}(x)$.

\begin{defn}
    A \emph{limited} ($\fun{sg}$-)\emph{polynomial expression} is an expression  over the signature $\{+, -, \div 2\}$ (and the $\fun{sg}$ function, resp.) on the set of variables/terms $X=\{x_1,\dots, x_h\}$ and integer constants; ($\fun{sg}$-)\emph{polynomial expressions} are defined analogously over the signature $\{ +, -, \div 2,\times\}$.  
\end{defn}

\noindent
To define the notion of degree for polynomial expressions, we follow the generalization to sets of variables introduced in~\cite{ADK24a}.

\begin{defn}
    Given a list of variables $\tu x=x_{i_1},\dots, x_{i_m}$, for $i_1,\dots, i_m \in \{1,\dots, h\}$, the \emph{degree of a set $\tu x$ in a $\fun{sg}$-polynomial expression $P$}, $\deg(\tu x,P)$, is inductively defined as follows:
    \begin{itemize}
        \itemsep0em
        \item $\deg(\tu x, x_{i_j})=1$, for $x_{i_j}\in \{x_{i_1}, \dots, x_{i_m}\}$, and $\deg(\tu x, x_{i_j})=0$, for $x_{i_j}\not\in \{x_{i_1},\dots, x_{i_m}\}$
        \item  $\deg(\tu x, \fun{sg}(P))=\deg(t\tu x, c) = 0$, for $c$ being an integer constant
        \item $\deg(\tu x, P+Q)=\deg(\tu x, P-Q)=\max\{\deg(\tu x, P), \deg(\tu x, Q)\}$
        \item $\deg(\tu x, P\times Q)=\deg(\tu x, P)+\deg(\tu x, Q)$ and $\deg(\tu x, P\div 2)=\deg(\tu x, P)$.
    \end{itemize}
\end{defn}

    %

\noindent
A $\fun{sg}$-polynomial expression $P$ is said to be \emph{essentially constant in a set of variables $\tu x$} when $\deg(\tu x, P)=0$
and \emph{essentially linear in $\tu x$}, when $\deg(\tu x, P)=1$. 

\begin{example}
    Let $\tu x = \{x_1, x_2, x_3\}$.
    The expression $P(x_1, x_2, x_3)=3x_1x_3 + 2x_2x_3$ is essentially linear in $x_1$, in $x_2$ and in $x_3$.
    It is not essentially linear in $\tu x$, as $\deg(P, \tu x)=2$.
    The expression $P'(x_1,x_2,x_3)=x_1\fun{sg}((x_1-x_3) \times x_2) + x_2^3$ is essentially linear in $x_1$, essentially constant in $x_3$ and not linear in $x_2$.
    Clearly, $P'$ is not linear in $\tu x$.
\end{example}

\noindent
The concept of \emph{linearity} allows us to control the growth of functions defined by ODEs.
In what follows, we will consider functions $\tu f:\Nat^{p+1}\to \mathbb{Z}^d$, i.e.~vectors of functions $\tu f=f_1, \dots, f_d$ from $\Nat^{p+1}$ to $\mathbb Z$, and introduce the linear $\ell$-ODE schema.

\begin{defn}[Linear $\lambda$-ODE]~\label{def:linear length ODE}
    Given $\tu g:\Nat^p \to \Nat^d, \tu h,\lambda: \Nat^{p+1} \to \mathbb{Z}$ and $\tu u:\mathbb{Z}^{d+1}\times \Nat^{p+1} \to \mathbb{Z}^d$, the function $\tu f:\Nat^{p+1}\to \mathbb{Z}^d$ is \emph{linear $\lambda$-ODE} definable from $\tu g, \tu h$ and $\tu u$ if it is the solution of the IVP with initial value $\tu f(0, \tu y) = \tu g(\tu y)$ and such that:
    $$
        \frac{\partial f(x, \tu y)}{\partial \lambda} = \tu u\big(\tu f(x, \tu y), \tu h(x, \tu y), x, \tu y\big)
    $$
    where $\tu u$ is \emph{essentially linear in the list of terms $\tu f(x, \tu y)$}.
    For $\lambda = \ell$, such schema is called \emph{linear length ODE}.
\end{defn}

\noindent 
If $\tu u$ is \emph{essentially linear in $\tu f(x, \tu y)$}, there exist matrices $\tu A$ and $\tu B$, whose coefficients are essentially constant in $\tu f(x, \tu y)$, such that $\tu f(0, \tu y)=\tu g(\tu y)$ and:
$$
\frac{\partial \tu f(x, \tu y)}{\partial \ell} = 
\tu A\big(x,\tu y, \tu h(x, \tu y),\tu f(x, \tu y)\big) \times \tu f(x, \tu y) + \tu B\big(x, \tu y, \tu h(x, \tu y), \tu f(x, \tu y)\big).
$$
Then, for all $x$ and $\tu y$, $\tu f(x, \tu y)$ is
$$
\sum^{\ell(x)-1}_{u=-1} \Bigg(\prod^{\ell(x)-1}_{t=u+1} \bigg(1+ \tu A\Big(\alpha(t), \tu y, \tu h(\alpha(t), \tu y), \tu f(\alpha(t), \tu y)\Big)\bigg) \Bigg) \times \tu B\big(\alpha(u), \tu y, \tu h(\alpha(u), \tu y), \tu f(\alpha (u), \tu y)\big)
$$
with the convention that $\prod^{x-1}_x\kappa(x)=1$ and $\tu (\alpha(-1), \tu y, \cdot, \cdot )=\tu f(0, \tu y)$.
Actually, the use of matrices is not necessary for the scope of this paper, but it is consistent with the general presentation of computation through essentially linear ODEs, of which full details can be found in~\cite{BournezDurand19,BournezDurand23}.

One of the main results of~\cite{BournezDurand19} is the implicit characterization of $\FP$ by the algebra made of basic functions $\fun{0}, \fun{1}, \pi^p_i, \ell, +, -, \times, \fun{sg}$, where $\pi^p_i$ denotes the projection function, and closed under composition ($\circ$) and $\ell$-ODE.

\subsection{On Small Circuit Classes and Function Algebras}

Boolean circuits are vertex-labeled directed acyclic graphs whose nodes/gates are either input nodes, 
output nodes
or are labeled with a Boolean function from the set $\{\wedge, \vee, \neg\}$.
Boolean circuit with modulo gates allows in addition gates labeled \textsc{Mod 2} or, more generally, 
\textsc{Mod $p$}, that output the sum of the inputs modulo $2$ or $p$, respectively.
Finally, Boolean circuit with majority gates includes \textsc{Maj} gates, that output 1 when the majority of the inputs are 1's. 
A family of circuits 
is $\mathbf{Dlogtime}$-uniform if there is a Turing machine (with a random access tape) that decides in deterministic logarithmic time the direct connection language of the circuit, i.e.~which, given 1$^n, a,b$ and $t\in \{\wedge, \vee, \neg\}$, decides if $a$ is of type $t$ and $b$ is a predecessor of $a$ in the circuit (and analogously for input/output nodes).
When dealing with circuits, the resources of interests are \emph{size}, i.e. the number of gates, and \emph{depth}, i.e.~the length of the longest path from input to output (see~\cite{Vollmer} for further details and related results).

\begin{defn}[Classes $\AC^i$ and $\NC^i$]
    For $i\in \Nat$, 
    $\mathbf{AC}^i$ (resp., $\mathbf{NC}^i$) is the class of language recognized by a $\mathbf{Dlogtime}$-uniform family of Boolean circuits of unbounded (resp. bounded) fan-in gates of polynomial size and depth $O(($log $n)^i)$.
    We denote by $\AC^i$ (resp., $\NC^i$) the corresponding function class.
\end{defn}
\noindent 
It is known that unbounded (poly$(n)$) fan-in can be simulated using a bounded tree of depth $O($log $n)$, so that $\NC^i \subseteq \AC^i \subseteq \NC^{i+1}$.
However, the inclusion has been proved to be strict only for $i=0$.
In contrast, when dealing with classes with counting ability, interesting bounds have been established. 

\begin{defn}[Classes $\mathbf{FACC[p]}$ and $\TC^0$]
    The class $\mathbf{ACC[p]}$ (resp., $\mathbf{TC}^0$) is the class of languages recognized by a $\mathbf{Dlogtime}$-uniform family of Boolean circuits including \textsc{Mod $p$} (resp., \textsc{Maj}) gates of polynomial size and constant depth.
    We denote by $\mathbf{FACC[p]}$ (resp., $\TC^0$) the corresponding function class.
\end{defn}
\noindent
We can even consider other levels in the $\TC$ hierarchy, such that for $i\in \Nat$, $\TC^i$ corresponds to classes of functions computed by families of unbounded fan-in circuits including \textsc{Maj} gates, of polynomial size and depth $O$((log $n)^i$).
It is known that the inclusion $\AC^0 \subset \FACC$ is proper, as the parity function cannot be computed in constant depth by standard unbounded fan-in Boolean gates, but it can be if \textsc{Mod 2} gates are added. 
Moreover, since any \textsc{Mod $p$} gate can be expressed by \textsc{Maj}, $\mathbf{FACC[p]}\subset \TC^0$.

    As mentioned, Cobham's work paved the way to several recursion-theoretic characterizations for classes other than $\FP$, including small circuit ones.
    In 1988/90, Clote introduced the algebra $\mathcal{A}_0$ to capture functions in the log-time hierarchy (equivalent to $\AC^0$) using so-called \emph{concatenation recursion on notation} (CRN)~\cite{Clote88,Clote1990}.
    The greater classes $\AC^i$ and $\NC$ are captured adding \emph{weak bounded recursion on notation}~\cite{Clote1990}, while
    $\NC^1$ has been characterized both in terms of the \emph{$k$-bounded recursion on notation} ($k$-BRN) schema~\cite{Clote93} and relying on the basic function $tree$~\cite{Clote1990}.
    In 1990, an alternative algebra (and logic), based on a schema called \emph{upward tree recursion}, was introduced to characterize $\NC^1$~\cite{ComptonLaflamme}.
    %
    In 1991, a new recursion-theoretic characterization (and arithmetic \emph{\`a la Buss}) for $\NC$ was presented~\cite{Allen}.
    In 1992 and 1995, other small circuit classes, including $\TC^0$, were considered (and re-considered) by Clote and Takeuti, who also introduced bounded theories for them~\cite{CloteTakeuti92,CloteTakeuti}.
    %
    %
    More recently, characterizations for $\NC^i$ were developed in~\cite{BonfanteKahleMarionOitavem}, this time following the ramification approach. 
    In~\cite{ADK24a}, the first ODE-based characterizations for $\AC^0$ and $\TC^0$ were introduced. 
    Related approaches to capture small circuit classes have also been provided in the framework of model- and proof-theory, both in the first and second order~\cite{Immerman,BarringtonImmermanStraubing,GurevichLewis,Lindell,ComptonLaflamme,Allen,CloteTakeuti92,CloteTakeuti,Johannsen,CookNguyen,Arai,CookMorioka,DurandHaakVollmer}.
    %
    %
    %
    %
    %
    To the best of our knowledge, so far, no \emph{uniform} implicit framework to characterize these classes has been offered in the literature.

\section{Towards a Uniform Understanding of Small Circuit Classes}\label{sec:towards}

We informally outline here the comprehensive investigation of small circuit classes in terms of ODEs (see Figure~\ref{fig} for a complete description). 
By slightly modifying the restrictions on the linearity allowed  in $\ell$-ODE and $\ell_2$-ODE, we can perform computation and, in several cases, characterize multiple function classes in the range from  $\mathbf{FAC}^0$  to  $\mathbf{FAC}^1$.

\begin{notation}\label{notation:strict}
    An equation or, by extension, an ODE schema is said to be \emph{strict} when it does not include any call to $f(x, \tu y)$ in $A$ or $B$ (not even under the scope of the sign function). 
    We say that $B$ includes \emph{simple} calls to $f(x,\tu y)$ if $f(x,\tu y)$ occurs in $B$ only in expressions of the form $s=\fun{sg}(f(x,\tu y))$.
    \\
    For clarity's sake, we will use $k(x, \tu y)$ for functions with limited range, typically in $\{0,1\}$, and $K(x, \tu y, h(x, \tu y), f(x, \tu y))$ for (possibly limited) $\fun{sg}$-polynomial expressions of restricted range.
    \\
    We use $s\big(x,\tu y, h(x, \tu y), f(x, \tu y)\big)$ or (when not ambiguous) simply $s(x)$ as an abbreviation for signed expressions in which $f(x, \tu y)$ can occur, i.e.~$s=\fun{sg}\big(E(x, \tu y, h(x, \tu y), f(x, \tu y))\big)$, where $E$ is any polynomial expression in which $x, \tu y, h(x, \tu y)$ and $f(x, \tu y)$ may appear.
    Finally, for readability, we often use 
    \begin{itemize}
        \itemsep0em
        \item $s(t)$ as a shorthand for $s\big(\alpha(t),\tu y, h(\alpha(t), \tu y), f(\alpha(t), \tu y)\big)$, 
        \item $K(t,\tu y, h, f)$, or even $K(t)$, as a shorthand for $K(\alpha(t), \tu y, h(\alpha(t),\tu y), f(\alpha(t), \tu y))$,
        \item $K_i(t)$ (or $B_i(t)$), with $i\in \{0,1\}$, as a shorthand for $K(\alpha(t), \tu y, h(\alpha(t), \tu y), f(\alpha(t), \tu y))$ in which $s(\alpha(t))=i$.
    \end{itemize}
    When different signed expressions $s_1, \dots, s_r$ occur in $A$ or $B$ these notational conventions are generalized accordingly.
    If not stated otherwise, we assume that expressions used in $\AC^0$ and $\FACC$ schemas are limited $\fun{sg}$-polynomial expressions (since multiplication is not in these classes), while those in $\NC^1$ and $\AC^1$ are $\fun{sg}$-polynomial expressions.
\end{notation}

Let us first examine the evolution in the expressive power of the linear schema (see Def.~\ref{def:linear length ODE}), which captures $\FP$,
$$A(x,\tu y, h, f) \times  f(x, \tu y) + B(x, \tu y,  h, f)$$

\noindent
when calls to $f$ are not allowed, i.e.~when dealing with equation of the form $A(x,\tu y, h) \times f(x, \tu y) + B(x, \tu y,  h)$, but still deriving along $\ell$. 
In this case, $f$ can be expressed as a combination of iterated multiplications and products and is in $\TC^0$. 
Interestingly, this schema is also the key ingredient to capture $\TC^0$ and to provide a corresponding completeness result. 
On the other hand, if calls to $f$ are allowed (i.e.~we are back to the general form), but the derivative is along the ``slower'' function $\ell_2$, then expressivity jumps to $\NC^1$. 

As a second example to focus on the smallest classes, consider schemas of the form $-f(x,\tu y)+B(x,\tu y,h,f)$, for $B\ge 0$, with derivative along $\ell$ and only \textit{simple} calls to $f$. 
We show that such defined functions are in $\NC^1$. 
Forcing the second part to be in $\{0,1\}$, i.e.~considering equations of the forms $-f(x,\tu y)+K(x,\tu y,h,f)$ and $-f(x,\tu y)+k(x,\tu y,h)$ (namely,~in the latter case, calls to $f$ are not allowed) makes the functions respectively in $\FACC$ and $\AC^0$. 
Remarkably, in the first two cases, the mentioned schemas are central to characterize the corresponding classes.

Finally, we examine (strict) equations of the form $B(x,\tu y,h)$ (where $B$ is a limited $\sign$-polynomial expression), this time deriving along $\loglog$. 
Equations of this form provide the ODE equivalent to log iterated sums and, not surprisingly, define a computation in $\AC^0$. 
Generalizing them to the non-strict setting by allowing equations of the form $B(x,\tu y,h,f)$, where calls to $f$ occur, gives a really more tricky schema, that is shown computable in $\TC^0$.


\begin{figure}\
\caption{ODE Schemas and Small Circuit Classes in a Nutshell}\label{fig}
\centering
\begin{center}
\begin{tabular}{c || c | c | c }
& strict schemas along $\ell$ & non-strict schemas along $\ell$ & schemas along $\ell_2$ \\
\hline
\hline
$\AC^0$ &
 $(k(x)-1) \times f(x)$ 
 & $(K(x,f)-1) \times f(x)$ 
 & $B(x)$  \\
  &
 $\pm \textcolor{blue}{f(x) + k(x)}$ 
 &  
 &  \\ 
 \hline
 $\FACC$ &
  $ -k(x) \times f(x) + k(x) \times k'(x)$
  & \textcolor{blue}{$-f(x) + K(x,f)$}  &  \\  
  %
  %
%
  \hline
 $\TC^0$ & \textcolor{blue}{$A(x)\times f(x)+B(x)$} & $K(x,f)$ & $B(x,f)$  \\
 & & $K(x,f) \times f(x)$ & \\ \hline
 $\NC^1$ & & $f(x) + K(x,f)$ & $A(x,f) \times f(x) + B(x,f)$\\
 & & \textcolor{blue}{$-f(x) + B(x,f)$} & \\
 \hline
 $\AC^1$ & & $B(x,f)$ & \\
\end{tabular}
\end{center}
\footnotesize
\raggedright
The table summarizes the features defining the equations of the main linear schemas in the corresponding small circuit class.
For readability, the arguments $h $ and $\tu y$ are omitted. 
Schemas in blue are not only computable in the corresponding class, but also provide a characterization for it.
Calls in non-strict schemas deriving along $\ell$ are simple or restricted to $s=\fun{sg}(f(x)-c)$, for $c\in \Nat$, except for the schema in $\FACC$, and $B$ is assumed to take positive values only.
\end{figure}
\normalsize

\section{Investigating Circuit Complexity Deriving along $\ell$}\label{sec:alongL}

In this section, we present linear ODE schemas obtained by deriving along $\ell$ and show that they are computable and, in some cases, capable of characterizing the corresponding class.

\subsection{The Class $\AC^0$}\label{sec:AC}

We start by revising the characterization of $\AC^0$ from~\cite{ADK24a}.
Recall that the function algebra $\ACDL$ capturing $\AC^0$ is based on the two very restrictive forms of $\ell$-ODE's below.

\begin{defn}[Schemas $\ell$-ODE$_1$, $\ell$-ODE$_3$]
    Given $g:\Nat^p\to \Nat$ and $k:\Nat^{p+1} \to \{0,1\}$, the function $f:\Nat^{p+1} \to \Nat$ is obtained by $\ell$-ODE$_1$ from $g$, if it is the solution of the IVP defined by  $f(0, \tu y)=g(\tu y)$ and 
    $$
    \frac{\partial f(x, \tu y)}{\partial \ell} = f(x, \tu y)+ k(x, \tu y).
    $$
    On the other hand, $f$ is obtained by $\ell$-ODE$_3$ from $g$ if it is the solution of the IVP defined by  $f(0, \tu y)=g(\tu y)$ and
    $$
    \frac{\partial f(x, \tu y)}{\partial \ell}=-\bigg\lceil \frac{f(x, \tu y)}{2}\bigg\rceil.
    $$
\end{defn}

\noindent
Relying on them, in~\cite{ADK24a}, the ODE-based function algebra below is introduced:
$$
\ACDL =
[\fun{0}, \fun{1}, \pi^p_i, \ell, +, -, \#, \div 2, \fun{sg}; \circ, \ell\text{-ODE}_1, \ell\text{-ODE}_3]
$$
where, as standard, $x \# y = 2^{\ell(x)\times \ell(y)}$.
In the same paper, this class is shown to precisely characterize $\AC^0$, see~\cite[Cor. 19]{ADK24a}.
%
Here, we dramatically simplify the (indirect) completeness proof, by showing that  $\ell$-ODE$_3$ precisely captures the idea of looking for the value of a bit in an $\ell(x)$ long sequence, while $\ell$-ODE$_1$ corresponds to the idea of left-shifting (the binary representation of) a given number, possibly adding 1.

\begin{remark}[The Function $\fun{BIT}$]\label{remark:BIT}
    Given an input $y$ of length 
    $\ell(y)$, for any $i\in \{0,\dots, \ell(y)\}$, the schema $\ell$-ODE$_3$ allows us to compute its $i^{th}$ bit in a direct way.
    Let us consider an instance of $\ell$-ODE$_3$ such that $g(\tu y)=y$, i.e. defined by $\fun{rsh}(0, y)=y$ and
    $$
    \frac{\partial \fun{rsh}(x,y)}{\partial \ell} = - \bigg\lceil \frac{ \fun{rsh}(x, y)}{2}\bigg\rceil.
    $$
    Then, the $\ell(x)^{th}$ bit of $y$ is $\fun{BIT}(x, y)=\fun{rsh}(x,y)-\fun{rsh}(x+1, y) \times 2$.
\end{remark}

\begin{remark}[The CRN Schema]\label{remark:CRN}
    In~\cite{Clote1990}, CRN is defined by $f(0, \tu y)=g(\tu y)$ and $f(s_i(x), \tu y)=s_{h_i(x,\tu y)}(f(x, \tu y))$, for $i\in \{0,1\}$, $h_i$ taking value in $\{0,1\}$ and $s_i(x)=2x+i$ being binary successor functions.
    We can rewrite this schema 
    as an instance of
    $\ell$-ODE$_1$ such that 
    \small
    \begin{align*}
        f_{crn}(0, y, \tu y) &= g(\tu y) \\
        \frac{\partial f_{crn}(x, y, \tu y)}{\partial \ell} &= f_{crn}(x,y,\tu y) + h_{\fun{BIT}(\ell(y)-\ell(x)-1,y)}(x, \tu y) \\
        &= f_{crn}(x,y,\tu y) +   h_0(x, \tu y) \times \fun{cosg}\big(\fun{BIT}(\ell(y)-\ell(x)-1, y)\big) \\
        & \quad \quad \quad \quad \quad \quad \; \; + h_1(x, \tu y) \times \fun{sg}\big(\fun{BIT}(\ell(y)-\ell(x)-1, y)\big)
    \end{align*}
    \normalsize
    where, by definition, for any $t$, $h_i(t,\tu y) \leq 1$.
    Then, $\fun{crn}(x,\tu y)=f_{crn}(x,x,\tu y)=f(x, \tu y)$.
\end{remark}

\begin{theorem}[\cite{ADK24a}]
    $\AC^0=\ACDL$.
\end{theorem}
\begin{proof}
$(\supseteq)$ By~\cite[Prop. 12, 17]{ADK24a}. 
$(\subseteq)$ The proof is indirect and passes through Clote's algebra
$$
\mathcal{A}_0 = [\fun{0}, \pi^p_i, \fun{s}_0, \fun{s}_1, \ell, \fun{BIT}, \#; \circ, \text{CRN}],
$$
which, in~\cite{Clote1990}, is proved to characterize $\AC^0$.
Indeed, its basic functions 0, $\pi^p_i, \ell, \#$ and composition are by design in $\ACDL$, $s_0,s_1$ can be rewritten in our setting almost for free, while $\fun{BIT}$ and CRN can be encoded as in Remarks~\ref{remark:BIT} and~\ref{remark:CRN}.
\end{proof}

\begin{toappendix}
\subsection{Strict Schemas in $\AC^0$}\label{appendix:AC0}


\begin{lemma}
    Let $k: \Nat^{p} \to \{0, 1\}$ be computable in $\AC^0$. 
    Then, the function $f: \Nat^{p+1} \to \Nat$ defined by $k$ as the solution of the IVP below:
        \begin{align*}
        f(0, \tu y) &= g(\tu y) \\
        \frac{\partial f(x, \tu y)}{\partial \ell} &= - f(x, \tu y) + k(x, \tu y)
        \end{align*}
    is in $\AC^0$.
\end{lemma}
\begin{proof}
    By definition of $\ell$-ODE, we simply have to compute $f(x, \tu y) = k(x, \tu y)$, which can be done in $\AC^0$, whenever $k(x, \tu y)$ can be computed in $\AC^0$.
\end{proof}

\begin{lemma}\label{lemma:ODEzero}
Let $g:\Nat^p \to \Nat$ and $k:\Nat^{p+1}\to \{0,1\}$ be computable in $\AC^0$.
Then, the function $f:\Nat^{p+1}\to \Nat$ defined by $k$ as the solution of the IVP below:
    \begin{align*}
        f(0, \tu y) &= g(\tu y) \\
        \frac{\partial f(x, \tu y)}{\partial \ell} &= \big(k(x, \tu y) -1\big) \times f(x, \tu y)
    \end{align*}
is in $\AC^0$.
\end{lemma}

%
\begin{proof}
    By definition of $\ell$-ODE:
    $$
        f(x, \tu y) = \prod^{\ell(x)-1}_{u=-1} k(\alpha(u), \tu y)
    $$
    with the convention that $k(\alpha(-1), \tu y)=g(\tu y)$.
    Clearly $f(x, \tu y)=g(\tu y)$ when $k(\alpha(t), \tu y)=1$ for all $t\in \{0, \dots, \ell(x)\}$ and 0 otherwise.
    By hypothesis, we can compute each $k(\alpha(t), \tu y)$ in parallel and ``iterated multiplication'' by 0 (i.e.,~cancellation) and 1 (i.e.,~no change) can be performed in $\AC^0$, i.e.~by a single layer with unbounded fan-in $\wedge$.
\end{proof}
%
%
        %

\end{toappendix}

In the context of a broader understanding of the hierarchy between classes, we have also considered non-strict $\ell$-ODE schemas that are computable in $\AC^0$. 
As anticipated, this investigation demonstrates a strong link between depth constraints over circuits and the syntactical constraints that define the corresponding linear length ODE.
Particularly interesting in this setting is the IVP defined by linear equations in which $A\in \{-1,1\}$ and $B=0$. 
As long as $f(x, \tu y)$ occurs in $A(=K-1)$ only in expressions of the form $s=\fun{sg}(f(x, \tu y))$, computation through this schema is proved to be in $\AC^0$.
Remarkably, computing the solution of such system corresponds to performing bounded search.

\begin{lemma}\label{lemma:ACODE}
    Let $g:\Nat^p\to \Nat$ and $h:\Nat^{p+1} \to \Nat$ be computable in $\AC^0$.
    Then, $f(x, \tu y)$ defined from $g$ and $h$ as the solution of the IVP made of $f(0, \tu y)=g(\tu y)$ and 
    $$
        \frac{\partial f(x, \tu y)}{\partial \ell} = \big(K(x, \tu y, h(x, \tu y), f(x, \tu y))-1\big) \times f(x, \tu y)
    $$
    where $K\in \{0,1\}$ is a limited $\fun{sg}$-polynomial expression and $f(x, \tu y)$ occurs in it only under the scope of the sign function in expressions of the form $s=\fun{sg}(f(x, \tu y))$, is in $\AC^0$.
\end{lemma}
\begin{proof}[Proof Sketch]
    By definition of $\ell$-ODE, $f(x, \tu y)=\prod^{\ell(x)-1}_{u=-1}K\big(\alpha(u), \tu y, h(\alpha(u), \tu y), f(\alpha(u), \tu y)\big)$, with the convention that $K(\alpha(-1), \cdot)=g(\tu y)$.
    Clearly, if there is a $t\in \{0,\dots,x\}$ such that $K(t)=0$, then $f(x, \tu y)=0$. 
    Therefore, we compute $f(x, \tu y)=\prod^{\ell(x)-1}_{u=-1}K_1(u)$ which is either 0, if $g(\tu y)=0$ or there is a $K_1(u)=0$, or $g(\tu y)$, otherwise.
    By hypothesis, we can compute in $\AC^0$ both $g(\tu y)$ and, for any $u\in \{0,\dots, \ell(x)-1\}$, $K_1(u)$. Thus, we conclude in one layer defined by an unbounded fan-in $\wedge$-gate.
\end{proof}
\noindent
This result can actually be strengthened by showing that it holds even for generalized expressions of the form $s = \fun{sg}(f(x, \tu y) - r)$ occurring in $K$.

To conclude the presentation of schemas for $\AC^0$, let us consider simple cases of bounded search encoded by the aforementioned ODE schemas.
(For further details, see Appendix~\ref{appendix:AC}.)

\begin{example}[BSearch]
    Let $R\subseteq \Nat^{p+1}$ and $h_R$ be its characteristic function.
    Then, for all $x$ and $\tu y$, 
    $(\forall z \leq \ell(x))R(z,\tu y)= \fun{sg}(f(x, \tu y))$,
    where $f$ is defined by $f(0, \tu y)= \fun{sg}(h_R(0, \tu y))$
    and
    $
    \frac{\partial f(x, \tu y)}{\partial \ell} = (h_R(x, \tu y) -1) \times f(x, \tu y),
    $
    which is clearly an (very limited) instance of the schema presented in Lemma~\ref{lemma:ACODE}.
    Another example of limited search consists in checking whether, at some point, $f(x, \tu y)$ reaches or not a boundary, say $r$.
    This can be expressed in a natural way by an instance of the (generalized) non-strict schema (in $\AC^0$) defined by the equation $\frac{\partial f(x, \tu y)}{\partial \ell} = \big(\fun{sg}(f(x, \tu y)-r)-1\big)\times f(x, \tu y)$.
\end{example}

\begin{toappendix}
\subsection{Non-Strict Schemas in $\AC^0$}\label{appendix:strictAC0}
%
%
            %

\begin{lemma}[Generalized Lemma~\ref{lemma:ACODE}]\label{lemma:ACODEg}
    Let $g:\Nat^p\to \Nat$ and $h:\Nat^{p+1} \to \Nat$ be computable in $\AC^0$.
    Then, $f(x, \tu y)$ defined from $g$ and $h$ as the solution of the IVP made of $f(0, \tu y)=g(\tu y)$ and 
    $$
        \frac{\partial f(x, \tu y)}{\partial \ell} = \big(K(x, \tu y, h(x, \tu y), f(x, \tu y))-1\big) \times f(x, \tu y)
    $$
    where $K\in \{0,1\}$ is a limited $\fun{sg}$-polynomial expression and $f(x, \tu y)$ occurs in it only under the scope of the sign function in expressions of the form $s=\fun{sg}(f(x, \tu y)-c)$ for any constant $c\in \Nat$, is in $\AC^0$.
\end{lemma}

\begin{proof}
    Let $s_j(u)= \fun{sg}\big(f(\alpha(u), \tu y)-c_j\big)$, $j\in \{1,\dots, h\}$ be the list of signed terms appearing in $K$ and $c=\max_j c_j$.
    Let's assume that the list can be ordered so that, for any $j', j''\in \{0,\dots,h\}$ and $j''>j'$, if $s_{j''}=1$, then $s_{j'}=1$ and vice versa, if $s_{j'}=0$, then $s_{j''}=0$.
    For readability, we use $K_{j'}(u)$ as a shorthand for $K_{i_0\dots i_{j'} i_{j'+1}\dots i_j}(u)= K_{\underbrace{1\dots 1}_{j' \text{ times}}\underbrace{0\dots 0}_{j-j' \text{ times}}}(u)$ where, for any $j''\in \{1,\dots, j\}$, $i_{j''}\in \{0,1\}$, corresponding to $K(x, \tu y, h(x, \tu y), f(x, \tu y))$ obtained by substituting $s_0, \dots, s_{j'}$ with 0 and $s_{j'+1},\dots, s_{j}$ with 1.
    %
    %
    By definition of $\ell$-ODE,
    $$
    f(x, \tu y) = \prod^{\ell(x)-1}_{u=-1} K\big(\alpha (u), \tu y, h(\alpha(u), \tu y), f(\alpha(u), \tu y)\big)
    $$
    with the convention that $K(\alpha(-1), \tu y, h(\alpha(u), \tu y), f(\alpha(u), \tu y))=g(\tu y)$.
    Notice that if there is a $t$ such that $K(t)= 0$, then $f(x, \tu y)=0$.
    Then, $f(x, \tu y)$ is either 0, when $g(\tu y)=0$ or there is a $t$ such that $K(t)=0$,
    or $g(\tu y)$, if, for all $t$, $K(t)=1$.
    Since, for any $t\in \{0,\dots, \ell(x)\}$, $f(\alpha(t),\tu y)$ can only take two different values, namely 0 or $g(\tu y)$, all the values of $K_j(t)$ (i.e.,~of conditional $f(\alpha(t+1), \tu y$) and the corresponding $s_j(t)$ can be computed independently from each other, in the beginning.
    Intuitively this is implemented in constant depth as follows:
    \begin{itemize}
    \itemsep0em
        \item In parallel, compute $g(\tu y)$ and, for any $u\in \{0,\dots, \ell(x)-1\}$, $K_{0}(t), \dots, K_{j}(t)$ and $s_0(g(\tu y)), \dots, s_j(g(\tu y))$.
        \item In parallel, for each $j'\in \{0,\dots, j\}$ and $t\in \{0,\dots, \ell(x)-1\}$, check the consistency between the initial value (i.e. what is the correct value of each $s_{j'}(g(\tu y))$) and that between all subsequent values of $K_{j'}(t), K_{j'}(t), \dots$.
        This can be done via unbounded fan-in $\wedge$-gate.
        \item In one layer, check whether the correct (i.e.~consistent with the initial value of $g(\tu y)$) sequence of values corresponds to strictly positive $K(t)$ only. This is just the (unbounded) disjunction between all the previous results.
    \end{itemize}
\end{proof}
In the restrictive case of $s=\fun{sg}(f(x, \tu y))$ it is immediate to see that this (non-strict) schema can be rewritten in terms of (strict) $\ell$-ODE$_1$ (offering an alternative indirect proof of Lemma~\ref{lemma:ACODE}).
Indeed, since $s=\fun{sg}(f(x, \tu y))$, $f(x, \tu y) \neq 0$ when $g(\tu y)>0$ and for any $u\in \{0, \dots, \ell(x)-1\}$, $K(u)=K_1(u)>0$.
This can be expressed via $\ell$-ODE$_1$ as follows:
\begin{align*}
    f'(0, \tu y) &= \fun{cosg}(g(\tu y)) \\
    \frac{\partial f'(x, \tu y)}{\partial \ell} &= f'(x, \tu y) \times \fun{cosg}\big(K_1(x)\big).
\end{align*}
Then, $\fun{bsearch}(x,x)=\fun{cosg}(f'(x,x))$.
Thus, also this schema allows us to perform computation corresponding to bounded search, and in fact it can be rewritten in terms of weaker schema introduced in Lemma~\ref{lemma:ODEzero}.

\begin{remark} 
In the limited case of $s=\fun{sg}(f(x,\tu y))$, the given schema is actually equivalent to the one presented in Lemma~\ref{lemma:ODEzero}:
\begin{itemize}
    \itemsep0em
    \item It is clear that the schema of Lemma~\ref{lemma:ODEzero} is nothing but a special case of the one of Lemma~\ref{lemma:ACODE}, such that $f(x, \tu y)$ does not occur in $K$.
    \item The converse direction also holds.
    By definition, 
    $$
    K\big(x, \tu y, h(x, \tu y), f(x, \tu y)\big) = K_0(x) \times \fun{cosg}\big(f(x, \tu y)\big) + K_1(x) \times \fun{sg}\big(f(x, \tu y)\big) 
    $$
    Then,
    \begin{align*}
        \frac{\partial f(x,\tu y)}{\partial \ell} &= \big(K(x, \tu y, h(x, \tu y), f(x, \tu y)) - 1\big) \times f(x, \tu y) \\
        &= \big(\big(K_0(x) \times \fun{cosg}(f(x, \tu y)) + K_1(x) \times \fun{sg}(f(x,\tu y))\big)-1\big) \times f(x, \tu y) \\
        %
        &= K_0(x) \times \fun{cosg}(f(x, \tu y)) \times f(x, \tu y) + K_1(x) \times \fun{sg}(f(x, \tu y))\times f(x, \tu y) - f(x, \tu y)
        \\
        &= K_1(x) \times f(x, \tu y) - f(x, \tu y) \\
        &= (K_1(x)-1) \times f(x, \tu y)
    \end{align*}
    where $K_1(x)$ does not include any occurrence of $s$, i.e. $K_1(x) = k(x, \tu y) \in \{0,1\}$, making the corresponding IVP an instance of the schema of Lemma~\ref{lemma:ODEzero}.
\end{itemize}
\end{remark}
\end{toappendix}

\subsection{The Class $\FACC$}

By weakening the constraints on linearity of the schemas which define $\AC^0$, we introduce an $\ell$-ODE the computation of which is provably in $\FACC$ and actually characterizes this class.
We consider a non-strict schema such that $A=-1$ and $B=K\in \{0,1\}$.
%
As we will see (Section~\ref{sec:chODE}), it is a special case of the schema we will introduce to capture $\NC^1$.
Together with the functions and schemas defining $\ACDL$, it is also enough to capture $\FACC$, thus providing the very first ODE-based characterization for this small circuit class  ``with counters''.

\begin{defn}[Schema $\schODE$]
    Given $g:\Nat^p \to \Nat$ and $h:\Nat^{p+1}\to \Nat$, the function $f:\Nat^{p+1} \to \Nat$ is defined by $\schODE$ if it is the solution of the IVP with $f(0, \tu y) = g(\tu y)$ and
    $$
    \frac{\partial f(x, \tu y)}{\partial \ell} = -f(x, \tu y) + K\big(x, \tu y, h(x, \tu y), f(x, \tu y)\big)
    $$
     where $K\in \{0,1\}$ is a limited $\fun{sg}$-polynomial expression and $f(x, \tu y)$ occurs in $K$ only under the scope of the sign function.
\end{defn}

Notably, this schema allows us to capture the computation performed by \textsc{Mod 2} gates and to rewrite the parity function in a very natural way.

\begin{example}[The Parity Function]\label{ex:parity}
      Let $z$ be a string of $\ell(z)=n$ bits.
      To count the number of 1's in $z$ we consider the following system defined by $ \fun{par}(0, y) = \fun{BIT}(0, y)$ and:
      \begin{align*}
          \frac{\partial \fun{par}(x, y)}{\partial \ell} &= - \fun{par}(x,  y) + \big(\fun{sg}(\fun{par}(x, y))\times \fun{cosg}(\fun{BIT}(\ell(x) + 1, y)) + \\
          &\quad \quad \quad \quad \quad \quad \quad \; \;
          \fun{cosg}(\fun{par}(x, y)) \times \fun{sg}(\fun{BIT}(\ell(x) + 1, y))\big)
      \end{align*}
      which is clearly an instance of $\schODE$.
      The parity function is now computed by $\fun{par}(z,z)$.
\end{example}
\noindent
Notice that this example shows that $\schODE$ really increases the power of $\ell$-ODE$_1$, as the parity function is known not to be in $\AC^0$ (see~\cite{DBLP:conf/focs/FurstSS81}).
%
This also indicates that $\schODE$ is actually a schema capturing the specific computational power corresponding to $\FACC$.
Moreover, it is shown that the desired closure property holds for it.

\begin{lemma}\label{lemma:schODE}
    If $f(x, \tu y)$ is defined by $\ell$\emph{-b}$_0$\emph{ODE} from functions in $\FACC$, then $f(x, \tu y)$ is in $\FACC$ as well.
\end{lemma}

\begin{proof}[Proof Sketch]
    For simplicity, let us consider the simple case of $f(x, \tu y)$ occurring only under the scope of the sign function in expressions of the form $s=\fun{sg}(f(x, \tu y))$.
    Relying on the fact that $K(x, \tu y, h(x, \tu y), f(x, \tu y))= K_1(x) \times \fun{sg}(f(x, \tu y)) + K_0(x) \times \fun{cosg}(f(x, \tu y))$, we can rewrite the equation defining $\schODE$ as follows: 
 %
  %
  \begin{align*} 
        \frac{\partial f(x,\tu y)}{\partial \ell} &= - f(x,\tu y)\times \big(K_0(x) =K_1(x)\big) +
        \big(K_0(x)=K_1(x)\big)\times \\
        & \quad \quad \quad \quad \quad \Big(\big((K_0(x)=1) \times \textsc{Mod 2} \Big(\sum^{\ell(x)}_{i=x}\big(K_0(i)=K_0(i+1)\big)\Big)=0\big) \\
        & \quad \quad \quad \quad + \big((K_0(x)=0) \times \textsc{Mod 2} \Big(\sum^{\ell(x)}_{i=x}\big(K_0(i)=K_0(i+1)\big)\Big)=1\big)\Big).
    \end{align*}
    \normalsize
    Intuitively, if $K_0(t)\neq K_1(t)$, no change occurs.
    In contrast, if $K_0(t)=K_1(t)$ we restart the computation (as the preceding value of the input must not be the last one for which $K_0=K_1$) by removing the previous value of $f(t,\tu y)$, and return 1 if either one of the following holds:
    \begin{itemize}
        \itemsep0em
        \item[i] $K_0(t)=K_1(t)=1$, i.e.~$K(t)$  ``generates'' a 1, and the number of changes of values (temporarily assuming that $t$ is the last value in $\{0,\dots, \ell(x)-1\}$ such that $K_0(t)=K_1(t)$) between subsequent values, i.e.~$K_0(t)$ and $K_0(t+1)$, is such that 1 (possibly alternating)  ``propagates'' until the end of the computation or 
        \item[ii] $K_0(t)=K_1(t)=0$, i.e.~$K(t)$  ``generates'' a 0, and the subsequent number of changes of values is such that 0 does not propagate to the final value, i.e.~\textsc{Mod 2}\Big($\sum^{\ell(x)-1}_{i=t}\big(K_0(t)=K_0(t+1)\big)\Big)=1$.
    \end{itemize}
    It is easy to see that the value computed in this way corresponds to the final value of the computation, i.e.~$f(x,\tu y)$ is given by $f(t, \tu y)$ such that for any $t'\in \{t,\dots, \ell(x)\}, K_0(t')\neq K_1(t')$.
    Clearly, this is a special case of $\schODE$ in which no call to $f(x,\tu y)$ occurs in $A$ or $B$, namely of the form $\frac{\partial f(x,\tu y)}{\partial \ell}=-k(x, \tu y) \times f(x,\tu y) + k(x, \tu y) \times k'(x,\tu y)$, for $k(x, \tu y),k'(x, \tu y)\in \{0,1\}$.
    Computation through this schema can be implemented by circuits of constant depth and polynomial size including \textsc{Mod 2} gates: 
    for any $t\in \{0,\dots, \ell(x)-1\}$, we can compute $k(\alpha(t), \tu y)$ and $k'(\alpha(t), \tu y)$ in parallel and check both their mutual consistency, i.e. whether $k(\alpha(t), \tu y)=k'(\alpha(t),\tu y)$, and the consistency between subsequent pairs, i.e.~whether $k(\alpha(t), \tu y)=k(\alpha(t+1), \tu y)$; then, using \textsc{Mod 2} gates, we can compute whether the number of changes of values is even or odd and conclude in constant depth via (unbounded fan-in) $\vee$-gate. 
    (For further details, see Appendix~\ref{appendix:FACC}).
\end{proof}


\begin{toappendix}

\subsection{Strict Schemas in $\FACC$}\label{appendix:ACC}

We consider a strict schema, which can be computed in $\FACC$.
It weakens the constraints on $A$ so that 
the coefficient of $f(x, \tu y)$ can now be 0, but its value is related to that of $B$:
by a special limited form of linearity, in which $A=-k(x,\tu y)\in \{0,1\}$ and $B=k(x, \tu y)\times k'(x, \tu y)\in \{0,1\}$, it is not possible that the coefficient of $f(x,\tu y)$ is 0, when $B\neq 0$ (\textsc{bcount} is not in $\FACC$).

\begin{lemma}\label{lemma:kkODE}
    Let $g:\Nat^p\to \{0,1\}$ and $k, k':\Nat^{p+1}\to \{0,1\}$ be computable in $\FACC$, then computation through the schema defined by the IVP below:
    \begin{align*}
        f(0, \tu y) &= g(\tu y) \\
        \frac{\partial f(x, \tu y)}{\partial \ell} &= - f(x, \tu y) \times k(x, \tu y) + k(x, \tu y) \times k'(x, \tu y)
    \end{align*}
    is in $\FACC$.
\end{lemma}

\begin{proof}
By definition of $\ell$-ODE,
$$
f(x, \tu y) = \sum^{\ell(x)-1}_{u=-1} \Bigg(\prod^{\ell(x)-1}_{t=u+1} \Big(1-k(\alpha(t), \tu y)\Big) \Bigg) \times k'(\alpha(u), \tu y) \times k(\alpha(u), \tu y).
$$
Notice that $f(x, \tu y)\in \{0,1\}$.
Clearly, $\prod^{\ell(x)-1}_{t'=u'+1}\big(1-k(\alpha(t'), \tu y)\big)=1$ when $k(\alpha(t'), \tu y)=0$ for all $t'\in \{u'+1, \dots, \ell(x)-1\}$ and $\prod^{\ell(x)-1}_{t'=u'+1}(1-k(\alpha(t'), \tu y))\times k'(\alpha(u'), \tu y) \times k(\alpha(u'), \tu y) =1$ when, in addition, $k'(\alpha(u')) \times k(\alpha(u'))=1$.
Then, for all $u''\in \{u'+1, \dots, \ell(x)-1\}$, $\prod^{\ell(x)-1}_{t=u''+1} (1-k(\alpha(t), \tu y)) \big) \times k(\alpha(u''), \tu y) \times k'(\alpha(u''), \tu y)=0$ (since $k(\alpha(u''),\tu y)=0$).
\\
We can compute this value using constant-depth circuits including 
\textsc{Mod 2} gates:
%
\begin{itemize}
    \itemsep0em
    \item In parallel, we compute $g(\tu y)$ and, for any $z\in \{0, \dots, \ell(x)-1\}$, the $\ell(x)$ pairs of values $k(\alpha(z))$ and $k'(\alpha(z'))$ (both in $\{0,1\}$).
    This can be done in $\FACC$ by hypothesis.
    \item Again in parallel, compute consistency between each pair, i.e. whether $k(\alpha(z)) = k'(\alpha(z))$ and between pairs of subsequent values, i.e.~whether $k(\alpha(z))=k(\alpha(z+1))$
    \item In the subsequent layer, we consider computation through $\ell(x)$ \textsc{Mod 2} gates ``counting'' whether the number of changes of values from the considered $z\in \{0, \dots, \ell(x)-1\}$ to $\ell(x)-1$ is even or odd.
    \item In constant depth, for any $z\in \{0,\dots, \ell(x)-1\}$, we consider the disjunction of:
    \begin{itemize}
        \itemsep0em
        \item the (unbounded fan-in) conjunction of the result of the corresponding consistency gate (i.e.~checking whether $k(\alpha(z), \tu y)=k'(\alpha(z), \tu y)$), of $k(\alpha(z), \tu y)$ and of the negation of the \textsc{Mod 2} gate for $z$ and
        \item the (unbounded fan-in) conjunction of the result of the corresponding consistency gate (i.e. checking whether $k(\alpha(z), \tu y)=k'(\alpha(z), \tu y)$), of the negation of $k(\alpha(z), \tu y)$ and of \textsc{Mod 2} gate for $z$.
    \end{itemize}
    \item We conclude in one layer computing the (unbounded fan-in) disjunction of all the $\ell(x)$ results of the previous gates (at most one of which can be 1).
\end{itemize}
\end{proof}


%
%
%

\subsection{Non-Strict Schemas in $\FACC$}\label{appendix:FACC}

\begin{proof}[Proof of Lemma~\ref{lemma:schODE}]
    Let us start with the simple case $s=\fun{sg}(f(x, \tu y))$.
    Since $K\in \{0,1\}$,
    $$
    K(x, \tu y, h(x, \tu y), f(x, \tu y)) = K_1(x) \times \fun{sg}(f(x,\tu y)) + K_0(x) \times \fun{cosg}(f(x, \tu y)).
    $$
    Relying on this fact, it is possible to rewrite $\schODE$ as the solution of the IVP with initial value $f(0, \tu y)=g(\tu y)$ and such that:
    \small
    \begin{align*}
        \frac{\partial f(x, \tu y)}{\partial \ell} &= f(x, \tu y) \times (K_0(x) \neq K_1(x) -1) \\
        &\quad + \Bigg(\bigg(\big(K_0(x) = K_1(x) = 1\big) \times \textsc{Mod 2}\bigg(\sum^{\ell(x)}_{i=x}\big(K_0(i) = K_0(i+1)\big)\bigg) = 0\bigg) \\
        &\quad + \bigg(\big(K_0(x) = K_1(x) = 0 \big) \times \textsc{Mod 2} \bigg(\sum^{\ell(x)}_{i=x}\big(K_0(i) = K_0(i+1)\big)\bigg)=1\bigg)\Bigg) \\
        &= -f(x, \tu y) \times \big(K_0(x) = K_1(x)\big) + \big(K_0(x) = K_1(x)\big) \\
        &\quad \times \Bigg(\bigg(\big(K_0(x) = 1 \big) \times \textsc{Mod 2}\bigg(\sum^{\ell(x)}_{i=x} \big(K_0(i) = K_0(i+1)\big)\bigg) = 0\bigg) \\
        & \quad + \bigg(\Big(K_0(x) = 0\Big) \times \textsc{Mod 2}\bigg(\sum^{\ell(x)}_{i=x}\big(K_0(i) = K_0(i+1)\big)\bigg) =1 \bigg)\Bigg)
    \end{align*}
    \normalsize
    Intuitively, if $K_0(t) \neq K_1(t)$, no change occurs, i.e.~$f(t+1)=f(t)$.
    On the contrary, if $K_0(x)=K_1(x)$ we ``remove'' previous $f(t,\tu x)$ and return 1 if either
    \begin{itemize}
        \itemsep0em
        \item $K_0(t)=K_1(t)=1$, i.e.~$K(t)$ ``generates'' a 1, and the number of subsequent changes of values (i.e.~$K_0(t') \neq K_0(t'+1)$) such that 1 (alternating) propagates to the final value (assuming $t$ is the last value in $\{0,\dots, \ell(x)-1\}$ such that $K_0(t)=K_1(t)$) or
        \item $K_0(t)=K_1(x)=0$, i.e.~$K(t)$ ``generates'' a 0, and the subsequent number of changes of values is such that 0 does not propagate to the final value,
    \end{itemize}
and 0 otherwise.
It is easy to see that the value computed in this way for $t\in \{0,\dots, \ell(x)-1\}$ corresponds to the final value of the computation only when for any subsequent $t'\in \{t,\dots, \ell(x)\}$, $K_0(t')\neq K_1(t')$. Otherwise, the final value is given by the last $t'$ such that $K_0(t')=K_1(t')$. 
    Since now no call to $f(x, \tu y)$ occurs in the lienar equation, this IVP comes out to be an instance of the schema proved in $\FACC$ in Lemma~\ref{lemma:kkODE}, such that, in particular,  $k(x, \tu y)$ is $K_0(x)=K_1(x)$ and, clearly, both $k,k'\in \{0,1\}$ and can be computed in $\FACC$ (as $\textsc{Mod 2} \sum^{\ell(x)}_{i=x}(K_0(i) = K_0(i+1))$ corresponds to computation through \textsc{Mod 2} gates.
    %

    As mentioned, this proof can be generalized to avoid the limitation $s=\fun{sg}(f(x,\tu y))$, allowing linear expression under the scope of the sign function.
    Let us consider the occurrence of one such expression in $K$, namely $s=\fun{sg}\big(E(x, \tu y, f(x,\tu y))\big)$.
    Since $K\in \{0,1\}$ (and $A=-1$), the proof is similar to the previous one except for the fact that, in this case, for any $t\in \{0, \dots, \ell(x)-1\}$, we compute (in parallel) even the sign of $s_0(t), s_1(t)$, and we search for the last $s_0(t)=s_1(t)$.
    Informally, the corresponding circuit is implemented as follows:
    \begin{itemize}
        \itemsep0em
        \item For any $t\in \{0, \dots, \ell(x)-1\}$, compute in parallel $K_0(t), K_1(t)$ and the corresponding $s_0(t), s_1(t)$ (where, as usual, $s_0=\fun{sg}(E(\alpha(t), \tu y,0))$ and $s_1=\fun{sg}(E(\alpha(t),\tu y,1))$, i.e.~0 or 1 substitutes $f(\alpha(t-1), \tu y) =K(t-1)$ in $s$).
        \item We search for the last $t\in \{0, \dots, \ell(x)-1\}$ such that $s_0(t)=s_1(t)$ (if for any $t\in \{0,\dots, \ell(x)-1\}$, $s_0(t)\neq s_1(t)$, we consider $g(\tu y)$).
        \item The construction is as before, but, here, we count whether the (number of) changes of values between $s_0(t')$ and $s_0(t'+1)$, occurring for $t'\in \{t, \dots, \ell(x)-1\}$, is even or odd.
        \item We select the correct final value, between $K_0(\ell(x)-1)$ and $K_1(\ell(x)-1)$ accordingly. 
    \end{itemize}
\end{proof}

\longv{
\marginpar{\textcolor{red}{A direct connection between $\ODEACC$ and $\schODE$ is missing}}
\begin{remark}[Relationship with $\ODEACC$ and $\schODE$]
    \textcolor{blue}{Both $\ODEACC$ and $\schODE$, restricted to occurrences of the form $s=\fun{sg}(f(x,\tu y))$, are special cases of the schema presented in Lemma~\ref{lemma:kkODE}.
    Indeed, any instance of $\ODEACC$
    can be seen as an instance of the given schema such that $k(x, \tu y)=2h(x, \tu y)$ and $k'(x, \tu y)=\frac{1}{2}$.
    On the other hand, $\schODE$ can be rewritten as an IVP such that $k(x, \tu y)= K_0(x)\times K_1(x)+ (1-K_0(x)) \times (1-K_1(x))$, while $k'(x, \tu y)$ is a convoluted expression involving $K_0(x)$ and $K_1(x)$ and described in Lemma~\ref{lemma:schODE}.}
\end{remark}
}

\end{toappendix}

Since $\schODE$ has been shown to capture computations performed by \textsc{Mod 2} gates, our characterization for $\FACC$ is obtained by simply endowing $\ACDL$ with it.

\begin{theorem}\label{theorem:FACC}
$\FACC
%
        = [\fun{0}, \fun{1}, \ell, \fun{sg}, +, -, \div 2, \pi^p_i, \#; \circ, \ell\emph{-ODE}_1, \ell\emph{-ODE}_3, \ell\emph{-b}_0\emph{ODE}]
 $
\end{theorem}

\begin{proof}
    $(\supseteq)$ Following the proofs~\cite[Prop. 12, 17]{ADK24a} plus Lemma~\ref{lemma:schODE} for $\schODE$. 
    $(\subseteq)$ Consequence of~\cite[Th. 18]{ADK24a} and of the fact that the parity function can be simulated by $\schODE$ (see Example~\ref{ex:parity}).
\end{proof}

\begin{toappendix}
\subsection{ODE-Based Characterizations of $\FACC$}

\begin{proof}[Alternative Proof of Theorem~\ref{theorem:FACC}.($\supseteq$)]
We provide an alternative, indirect proof of $\supseteq$.
The proof relies on Clote and Takeuti's characterization of $\FACC$ in terms of BRN~\cite{CloteTakeuti}:
$$
\FACC = [0, I, s_0, s_1, \ell, BIT, \#; \circ, \text{CRN}, \text{1-BRN}]
$$
As said, all basic functions, $\circ$ and CRN can be rewritten in $\mathbb{ACDL}$, see~\cite{ADK24a}.
It remains to show that a function $f(x, \tu y)$ defined by 1-BRN from $g$ and $h$ as follows:
\begin{align*}
    f(0, \tu y) &= g(\tu y) \\
    f(x, \tu y) &= h\big(x, \tu y, f(\ell(x), \tu y)\big)
\end{align*}
with $f(x, \tu y)\leq 1$, can be rewritten as the solution of the IVP below:
\begin{align*}
    \fun{1brn}(0, \tu y) &= g(\tu y) \\
    \frac{\partial \fun{1brn}(x, \tu y)}{\partial \ell} &= - \fun{1brn}(x, \tu y) + \fun{sg}\big(f(x, \tu y)\big) \times h\big(\ell(x), \tu y, 1\big) + \fun{cosg}\big(f(x, \tu y)\big) \times h\big(\ell(x), \tu y, 0\big)
\end{align*}
which is an instance of (limited) $\schODE$.
\end{proof}

\longv{
\begin{theorem}
$\FACC
= 
[\fun{0}, \fun{1}, \ell, \fun{sg}, +, -, \div 2, \pi^p_i, \#; \circ, \ell\emph{-ODE}_1, \ell\emph{-ODE}_3, \ODEACC]
$
%
\end{theorem}

\begin{proof}
    $\subseteq$ Following the closure proofs from~\cite{ADK24a} plus Lemma~\ref{lemma:FACCone} and Lemma~\ref{lemma:schODE} for $\ODEACC$ and $\schODE$, respectively.

    $\supseteq$ Consequence of the fact that (unbounded fan-in) \textsc{Mod 2} computation or parity can be simulated by $\ODEACC$ (see Remark~\ref{remark:mod2}).
\end{proof}
}
\end{toappendix}

\subsection{The Class $\TC^0$}\label{sec:TC}
In this section, we introduce schemas the computation of which is proved to be in $\TC^0$.
One is defined in terms of strict linear $\ell$-ODE and provides a new, compact characterization for this class.
The second schemas considered are not strict, as they include simple calls to $f(x,\tu y)$, and are defined uniformly with respect to corresponding (i.e.,~non-strict) schemas for $\AC^0$ and $\FACC$ (as well as for $\NC^1$ and $\AC^1$, see  Sections~\ref{sec:chODE} and~\ref{sec:ACone})

We start by dealing with the linear ODE schema defined deriving along $\ell$ and so that $\sA(x, \tu y)$ and $\sB (x, \tu y)$ are both  $\fun{sg}$-polynomial expressions, with no call to $f(x, \tu y)$.

\begin{defn}[Schema $\pODE$]
    Given $g: \Nat^p \to \Nat$ and $h:\Nat^{p+1} \to \Nat$, the function $f:\Nat^{p+1}\to \Nat$ is defined by $\pODE$ from $g$ and $h$, if it is the solution of the IVP with initial value $f(0, \tu y) = g(\tu y)$  and such that
    $$
    \frac{\partial f(x, \tu y)}{\partial \ell} = \sA \big(x, \tu y, h(x, \tu y)\big) \times f(x, \tu y) + \sB\big(x, \tu y, h(x, \tu y)\big)
    $$
    where $\sA$ and $\sB$ are $\fun{sg}$-polynomial expressions.
\end{defn}
\noindent
This compactly encapsulates not only iterated multiplication and addition but also $\ell$-ODE$_1$ and $\ell$-ODE$_3$ computation.

\begin{example}[Counting and Iterated Addition in $\pODE$]\label{ex:pODE}
    Notably, \textsc{BCount} (resp. iterated addition, \textsc{ItAdd}) can be expressed as a special case of $\pODE$ such that $A=0$ and $B=k(x, \tu y)$ with $k,g\leq 1$, i.e.~by considering an IVP such that $f(0, \tu y) = g(\tu y)$ and:
    \begin{align*}
        \frac{\partial f(x, \tu y)}{\partial \ell} = k(x, \tu y) \qquad (\mbox{resp. }  \frac{\partial f(x, \tu y)}{\partial \ell} &= \sB(x, \tu y, h(x, \tu y)).
    \end{align*}
    Taking $g(\tu y)=\fun{BIT}(0, y)$ and $k(x, \tu y)=\BIT(\ell(x),y)$, $f(x,x)=\sharp \{u: \BIT(u,x)=1\}$. 
\end{example}

\noindent
These examples are especially relevant as they reveal that $\pODE$ is really more expressive than all strict schemas introduced so far (\textsc{BCount} is not computable in $\AC^0$ and $\FACC$, see~\cite{Razborov87,DBLP:conf/stoc/Smolensky87}).
The closure property can be shown to hold for $\pODE$ by relying on the following result.

\begin{prop}\label{prop:pODE}
 If $g: \Nat^p \to \Nat$ and $h:\Nat^{p+1} \to \Nat$ are computable in $\TC^0$ and $f:\Nat^{p+1}\to \Nat$ is defined by $\pODE$ from $g$ and $h$, then $f$ is computable in $\TC^0$ as well.
\end{prop}

\begin{proof} By definition of $\ell$-ODE, such $f$ satisfies
$f(0, \tu y)=g(\tu y)$ and
 $$
        f(x, \tu y) = \sum^{\ell(x)-1}_{u=-1} \Bigg(\prod^{\ell(x)-1}_{t=u+1} \bigg(1+\sA(r, \tu y)\bigg)\Bigg) \times \sB (u, \tu y)
    $$
   
\noindent with the convention that $\prod^{z-1}_z \kappa(x)=1$ and $\sB(-1, \tu y)=g(\tu y)$. Computing this expression can be done combining iterated addition and multiplication, which are in $\TC^0$.  
\end{proof}

\noindent
As an immediate byproduct we obtain a new, concise ODE-based characterization of $\TC^0$.

\begin{theorem}
    $\TC^0 = [\fun{0}, \fun{1}, \ell, \fun{sg}, +, -, \div 2, \pi^p_i, \#; \circ, \ell\emph{-pODE}].$
\end{theorem}
\begin{proof}
\begin{sloppypar}
    $(\supseteq)$ By Proposition~\ref{prop:pODE}. 
    $(\subseteq)$ Consequence of the facts that (i) $\ell$-pODE is strong enough to express $\ell$-ODE$_1$ and $\ell$-ODE$_3$,
    (see Appendix~\ref{appendix:strictTC}), 
    together with $\AC^0=\mathcal{A}_0=\ACDL$ (see~\cite{ADK24a,Clote1990}), and (ii) $\ell$-pODE captures \textsc{BCount} computation (see Example~\ref{ex:pODE}), which is equivalent to \textsc{Maj} or $\times$ (see~\cite{Vollmer}), together with the fact that $[\fun{0}, \pi^p_i, \fun{s}_0, \fun{s}_1, \ell, \fun{BIT} \times, \#; \circ, \text{CRN}]=\TC^0$ (see~\cite{CloteTakeuti}).
established by~\cite{CloteTakeuti}.
\end{sloppypar}
\end{proof}

\begin{toappendix} 
\subsection{Strict Schemas in $\TC^0$}\label{appendix:strictTC}
As mentioned the closure property can be shown for the linear ODE schema defined deriving along $\ell$ and considering $\sA(x, \tu y)$ and $\sB (x, \tu y)$ to be both polynomial expressions.
This is obtained relying on the fact that iterated multiplication and addition can be done in $\TC^0$.

\begin{lemma}[Strict Sharply Bounded Sum]\label{lemma:BSum}
    Let $\sB$ be a $\fun{sg}$ polynomial expression and $g:\Nat^p \to \Nat$ and $h: \Nat^{p+1}\to \Nat$ be functions computable in $\TC^0$.
    Then, the function $f:\Nat^{p+1}\to \Nat$ defined by $h$ and $g$ as the solution of the IVP made of $f(x, \tu y) = g(\tu y)$ and
    $$
        \frac{\partial f(x,\tu y)}{\partial \ell} = \sB(x, \tu y, h(x, \tu y)\big)
    $$
    is in $\TC^0$ as well.
\end{lemma}
\begin{proof}[Proof Sketch]
    By definition of $\ell$-ODE, the solution of this system is of the form 
    $$
    f(x, \tu y) = \sum^{\ell(x)-1}_{u=-1} \sB(\alpha(u), \tu y, h(\alpha(u), \tu y))
    $$
    where $\alpha(u)=2^u-1$ and with the convention that $\sB(\cdot, \alpha(-1), \tu y)=f(x, \tu y)=g(\tu y)$, and the fact that iterated addition is in $\TC^0$ (see~\cite{Vollmer}). 
\end{proof}


\begin{lemma}[Strict Sharply Bounded Product]\label{lemma:BProd}
    Let $\sA$ be a $\fun{sg}$ polynomial expression and $g:\Nat^{p}\to \Nat$ and $h:\Nat^{p+1}\to \Nat$ be functions computable in $\TC^0$.
    Then, the function $f(x, \tu y)$ defined from $h(x, \tu y)$ and $g(x, \tu y)$ as the solution of the IVP made of $f(0, \tu y)=g(\tu y)$ and
    $$
        \frac{\partial f(x, \tu y)}{\partial \ell} = \sA\big(x, \tu y, h(x, \tu y)\big)\times f(x, \tu y)
    $$
    is in $\TC^0$.
\end{lemma}
\begin{proof}[Proof Sketch]
    By definition of $\ell$-ODE, the solution of the system is of the form 
    $$
    f(x, \tu y) =\prod^{\ell(x)-1}_{u=-1} (1 + \sA(\alpha(u), \tu y, h(\alpha(u), \tu y))
    $$
    with the convention that $\prod^{x}_{u=-1}\kappa(x)=f(0, \tu y)=g(\tu y)$.
    The proof is analogous to the one for Lemma~\ref{lemma:BSum}, here remarking that iterated multiplication is in $\TC^0$ (see~\cite{HAB}).
\end{proof}

\noindent 
The two results above also hold when dealing with vectors rather than functions.

\begin{prop}\label{prop:aux}
    Let $\sA(x, \tu y)$ and $\sB(x, \tu y)$ be polynomial expressions.
    Let $z\in \Nat$ be such that $z\leq log^c(x)$, for some $c\in \Nat$ and 
    $$
        f(x, \tu y) = \sum^{z-1}_{u=-1} \Bigg(\prod^{z-1}_{t=u+1} \bigg(1+\sA(r, \tu y)\bigg)\Bigg) \times \sB (u, \tu y)
    $$
    with the convention that $\prod^{z-1}_z \kappa(x)=1$ and $\sB(-1, \tu y)=T(0, \tu y)$.
    Then, the function which, given $x, \tu y$ and $f(x, \tu y)$ as input, outputs $f(x, \tu y)$ can be computed in $\TC^0$.
\end{prop}

%



\begin{remark}[Relationship with $\ell$-ODE$_1$ and $\ell$-ODE$_3$]\label{remark:pODE}
    Observe that $\ell$-ODE$_1$ and $\ell$-ODE$_3$ are nothing but special cases of $\pODE$:
    $\ell$-ODE$_1$ is nothing but an instance of $\pODE$ such that $A=1$ and $B=k(x, \tu y)\in \{0,1\}$ and $\ell$-ODE$_3$ is a special case of $\pODE$ such that $A=(\div 2) -1$ and $B=0$. 
\end{remark}
\noindent
In other words, $\pODE$ encapsulates the computing power of $\ell$-ODE$_1$ and $\ell$-ODE$_3$ simultaneously.
As an immediate consequence $\AC^0 \subseteq \TC^0$.

\longv{
\begin{remark}[Generalized $\pODE$ in $\NC^1$]
    Actually, we can generalize this schema so that $\sA$ and $\sB$ are expressions defined by any $\AC^0$ functions and, since $f(x, \tu y)$ does not occur in them, even bounded sums is allowed.
\end{remark}
}


\end{toappendix}

Additionally, in order to better and uniformly understand upper bounds in circuit computation, we consider natural constraints over the non-strict linear length-ODE schema obtained by allowing \emph{simple} calls of $f(x, \tu y)$.
As shown throughout the paper, this kind of feature may increase the computational power of a schema, and, here, to ensure that computation remains in $\TC^0$ we have to restrict the forms of $A$ and $B$. 

\begin{lemma}\label{lemma:TCODE}
    Let $g:\Nat^{p}\to\Nat$ and $h:\Nat^{p+1}\to \Nat$ be in $\TC^0$. 
    Then, a function $f$ defined by the IVP with initial value $f(0, \tu y) = g(\tu y)$ and equation of one of the following forms: 
$$
        \frac{\partial f(x, \tu y)}{\partial \ell} = K\big(x, \tu y, h(x, \tu y), f(x, \tu y)\big) \mbox{ or }  \frac{\partial f(x, \tu y)}{\partial \ell} = K\big(x, \tu y, h(x, \tu y), f(x, \tu y)\big)\times f(x,\tu y)
$$

    \noindent where $K\in \{0,1\}$ is a $\fun{sg}$-polynomial expression and $f(x, \tu y)$ only occurs in expressions of the form $s=\fun{sg}(f(x, \tu y))$ in it, is in $\TC^0$ as well.
\end{lemma}

\begin{proof}
For the first equation type, notice that, if $g(\tu y)\ge 1$, then, for any $t\in \{0,\dots, x\}$, $f(t, \tu y)\ge 0$ (by definition of the schema, in which $A=0$ and $B\in\{0,1\}$).
Thus, in this case,
$
f(x, \tu y) = \sum^{\ell(x)-1}_{u=-1} K_1(u)
$
with the convention that $K_1(-1)=g(\tu y)$, and we can conclude by noticing that $\textsc{BCount}$ is in $\TC^0$.
Otherwise, we compute in parallel $g(\tu y)$ and $K_0(t), K_1(t)$ for all $t\in \{0, \dots, \ell(x)\}$.
Then, we search for the first $t$ such that $K_0(t)\neq 0$ and compute $\sum^{\ell(x)-1}_{u=t} K_1(u)$:
this would be the expected result since, before $t$, each summand is equal to 0 and, after $t$, all signed terms remain strictly positive.
Again, this corresponds to \textsc{BCount}, which can be done in $\TC^0$.
In terms of circuit computation, this is implemented as follows:
\begin{itemize}
    \itemsep0em
    \item In parallel, the circuit computes the value of $g(\tu y)$ and, for each $u\in \{0, \dots, \ell(x)-1\}$ of both $K_0(u)$ and $K_1(u)$ (this can be done in $\TC^0$ for hypothesis).
    \item In one layer, for any $t\in \{0, \dots, \ell(x)-1\}$, it computes the sum $K_0(t) + \sum^{\ell(x)-1}_{u=t+1}K_1(u)$ by $t$ \textsc{Count} gate.
    \item In constant depth, for any $t\in \{0, \dots, \ell(x)-1\}$, it computes the conjunction between  $\neg K_0(0), \neg K_0(1), ..., $ $\neg K_0(t-1)$ and $K_0(t)$ and select the corresponding \textsc{Count} gate.
\end{itemize}

For the second equation type, first note that, if $g(\tu y)=0$, trivially $f(x, \tu y)=0$ (which is clearly computable in $\TC^0$). Let $g(\tu y)\ge 1$. 
Then, for any $t\in \{0, \dots, \ell(x)-1\}$, $\fun{sg}(f(0, \tu y))>0$.
    Therefore, by definition of $\ell$-ODE,   
    $
        f(x, \tu y) = \prod^{\ell(x)-1}_{u=-1} \big(1+ K_1(u)\big),
    $
    with $\prod^{t'}_{t=-1} \kappa=g(\tu y)$.
    This can be computed in $\TC^0$: $g(\tu y)$ and each $K_1(u')$, for $u'\in \{0,\dots, \ell(x)-1\}$ can be computed in parallel in the beginning ($g(x, \tu y)$ and $h(x, \tu y)$ being computable in $\TC^0$ for hypothesis and $K_1$ being a $\fun{sg}$-polynomial expression by definition) and \textsc{ItMult} (so, in particular, iterated multiplication by 1 and 2) is in $\TC^0$, see~\cite{HAB}.
\end{proof}



\subsection{The Class $\NC^1$}\label{sec:chODE}
In this section, we consider two restrictions on $\ell$-ODE, which give us natural schemas the solution of which is a function computable in $\NC^1$: namely,~the schema defined by $A=1$ and $B=K\in \{0,1\}$ and the one defined by $A=-1$  and more liberal $B\ge 0$ (obtaining the so-called $\chODE$ schema, somewhat generalizing $\schODE$).

\begin{lemma}[Bounded Concatenation with Simple Calls]\label{lemma:ODENC1}
    Let $g:\Nat^p \to \Nat$ and $h: \Nat^{p+1} \to \Nat$ be computable in $\NC^1$.
    Then, computation through the schema defined by the IVP with initial value $f(0, \tu y)=g(\tu y)$ and such that:
    $$
    \frac{\partial f(x, \tu y)}{\partial \ell} = f(x, \tu y) + K\big(x, \tu y, h(x, \tu y), f(x, \tu y)\big)
    $$
    where $K\in \{0,1\}$ is a $\fun{sg}$-polynomial expression, and $f(x, \tu y)$ appears in it only under the scope of the sign function in expressions of the form $s=\fun{sg}(f(x, \tu y))$, is in $\NC^1$.
\end{lemma}

\begin{proof}[Proof Sketch]
\begin{sloppypar}
By definition of $\ell$-ODE,  $f(x, \tu y)=\sum^{\ell(x)-1}_{u=-1} 2^{\ell(x)-u-1} \times$ $K(\alpha(u),\tu y, h(\alpha(u), \tu y),$ $f(\alpha(u), \tu y))$, with the convention that $K(\alpha(-1))=g(\tu y)$.
Let us define the truncated sum (or concatenation) between $a$ and $b-1$ bounded by $\ell(x)-1$, as $S(a,b) = \sum^{b-1}_{u=a} K\big(\alpha(u), \tu y, h(\alpha(u), \tu y), (u)\big)$, with $S(u)=S(\alpha(u), \tu y)=f(\alpha(u), \tu y)$.
In particular, we consider two sums in which the sign of the first term is assumed to be either 0 or 1: $S_0(a,b)= \sum^{b-1}_{u=a} K(\alpha(u), \tu y, h(\alpha(u), \tu y), S(u))$, with $S_0(a)=0$ and $S_1(a,b)=\sum^{b-1}_{u=a} 2^{\ell(x)-u-1} \times K_1(u)$.
%
%
Therefore, for any $x$ and $\tu y$:
\end{sloppypar}
$$
f(x, \tu y) 
= S(-1, \ell(x)/2)+ S(\ell(x)/2, \ell(x)) 
= + \begin{cases} S(-1,\ell(x)/2) \\
S_0(\ell(x)/2, \ell(x)) \times (1-s(\ell(x)/2-1)) \\
S_1(\ell(x)/2, \ell(x)) \times s(\ell(x)/2-1). \end{cases}
$$
Clearly, $S(-1, \ell(x)/2), S_0(\ell(x)/2, \ell(x))$ and $S_1(\ell(x)/2, \ell(x))$ can be computed in parallel:
we know the actual starting value of $S_0(\ell(x)/2, \ell(x))$, which is 0, while to compute $S_1(\ell(x)/2, \ell(x))= \sum^{\ell(x)-1}_{u=\ell(x)/2}2^{\ell(x)-u-1}K_1(u)$ it is enough to know that, for any $a>\ell(x)/2$, $f(\alpha(a), \tu y)>0$.
Once $S(-1,\ell(x)/2)=f(\alpha(\ell(x)/2), \tu y)$ is computed and its sign tested, we can select between $S_0$ and $S_1$ accordingly.
We conclude by taking the concatenation between $S(-1,\ell(x)/2)$ and either $S_0$ or $S_1$.
\end{proof}

\noindent
This schema  can be generalized so to allow calls to $f(x, \tu y)$ in expressions of the form $s=\fun{sg}(f(x, \tu y)-c)$, i.e.~testing whether $f(x,\tu y)\geq c$.
Computation through an extended version in which $f(x, \tu y)$ occurs under the scope of $\fun{bit}$ is provably in $\NC^1$ 
but it is also strong enough to capture binary tree structure, hence to provide a direct encoding of $\NC^1$ circuit computation.

\begin{remark}
    The condition that $f$ appears only in expressions of the form $s=\fun{sg}(f(x, \tu y))$ or even  $\fun{sg}(f(x, \tu y)-c)$ might seem limited. However the possibilities it offers really increase the ODE's expressive power and, on the other side, even for slightly more liberal conditions, provide the only known upper bound for the schema is $\FP$. 
\end{remark}

\begin{toappendix}
\subsection{Schemas in $\NC^1$}\label{appendixNC}
\begin{proof}[Proof of Lemma~\ref{lemma:ODENC1}]
    By definition of $\ell$-ODE,
    $$
        f(x, \tu y) = \sum^{\ell(x)-1}_{u=-1} 2^{\ell(x)-u-1} \times K\big(\alpha(u), \tu y, h(\alpha(u), \tu y), f(\alpha(u), \tu y)\big)
    $$
    with the convention that $K(\alpha(-1), \tu y, h(\alpha(u), \tu y), f(\alpha(u), \tu y))=f(0, \tu y)=g(\tu y)$.
    Let us consider the truncated sum (or concatenation) between $a$ and $b-1$ bounded by $\ell(x)-1$, which is defined as:
    $$
    S(a,b) = \sum^{b-1}_{u=a} K\big(\alpha(u), \tu y, h(\alpha(u), \tu y), S(u), \tu y)\big) = \sum^{b-1}_{u=a} K\big(\alpha(u), \tu y, h(\alpha(u), \tu y), s(u)\big)
    $$
    with $S(a, \tu y)=f(\alpha(a), \tu y)$ and $s(a) = \fun{sg}\big(f(\alpha(a), \tu y)\big)$.
    In particular, we consider the following two sums in which the sign of the first term is assumed to be either 0 or 1:
    \begin{align*}
        S_0(a, b) &= \sum^{b-1}_{u=a} K(\alpha(u), \tu y, h(\alpha(u), \tu y), S_0(u)) \\
        S_1(a, b) &= \sum^{b-1}_{u=a} 2^{\ell(x)-u-1} K_1(u)
    \end{align*}
    with $S_0(u)=0$. 
    Clearly, if $K(a)=K_1(a)$, $f(\alpha(a), \tu y)$ must be strictly greater than 0 and, then, by definition, for any $t\in \{\alpha(a),\dots, x\}$, $f(t, \tu y)>0$.
    Therefore, for any $x$ and $\tu y$, the value of $f(x, \tu y)$ can be computed as follows:
    \begin{align*}
        f(x, \tu y) &= S(-1, \ell(x)) \\
        &= S(-1, \ell(x)/2) + S(\ell(x)/2, \ell(x)) \\
        &= f(\alpha(\ell(x)/2), \tu y) + S(\ell(x)/2, \ell(x)) \\
        &= \begin{cases}
            0 + S_0(\ell(x)/2, \ell(x)) \\
            S(-1, \ell(x)/2) + \sum^{b-1}_{u=a} 2^{\ell(x)-u-1} K_1(u)
        \end{cases}
    \end{align*}
    Hence,
    $$
     f(x, \tu y) = + \begin{cases}
        S(-1, \ell(x)/2) \\
        S_0(\ell(x)/2, \ell(x)) \times 1 - s(\ell(x)/2-1) \\
        S_1(\ell(x)/2, \ell(x)) \times s(\ell(x)/2-1)
    \end{cases}
    $$
    Independently computing $S(-1, \ell(x)/2), S_0(\ell(x)/2, \ell(x))$ and $S_1(\ell(x)/2), \ell(x))$ can be done in parallel: as seen, we know the actual starting value of $S_0(\ell(x)/2, \ell(x))$, which is $f(\alpha(\ell(x)/2), \tu y)=0$ (so that this sum can be computed exactly as the starting sum $S(-1, \ell(x)/2)$ in at most $\ell(x)$ steps), while to compute $S_1(\ell(x)/2, \ell(x))=\sum^{\ell(x)-1}_{u=\ell(x)/2} 2^{\ell(x)-u-1}K_1(u)$ we do not need to know the actual value of $f(\alpha(\ell(x)/2), \tu y)$ since we know that it is always strictly greater than 0, i.e. for $t\in \{a, \dots, x\}$ $K(t)=K_1(t)$.
    Once that $S(-1, \ell(x)/2)=f(\alpha(\ell(x)/2), \tu y)$ is known, the sign can be tested and selection between $S_0$ and $S_1$ obtained accordingly.
    The final step amounts in a concatenation between $S(-1, \ell(x)/2)$, which is either 0 or has length strictly smaller than $\ell(x)/2$, and $S_0$ or $S_1$.
\end{proof}

\begin{lemma}[Generalized Bounded Concatenation with Simple Calls]\label{lemma:gBC}
    Let $g:\Nat^p\to \Nat$ and $h:\Nat^{p+1}\to \Nat$ be functions computable in $\NC^1$.
    Then, $f:\Nat^{p+1}\to \Nat$ defined by $h$ and $g$ as follows:
    \begin{align*}
        f(0, \tu y) &= g(\tu y) \\
        \frac{\partial f(x, \tu y)}{\partial \ell} &= f(x, \tu y) + K\big(x, \tu y, h(x, \tu y), f(x, \tu y)\big)
    \end{align*}
    with $K\in \{0,1\}$ and $f$ appearing only under the scope of the sign function inside expressions of the form $s=\fun{sg}(f(x,\tu y) - c)$ for any constant $c\in \Nat$ is in $\NC^1$.
\end{lemma}
\begin{proof}
    By definition of $\ell$-ODE,
    $$
    f(x, \tu y) = \sum^{\ell(x)-1}_{u=-1} 2^{\ell(x)-u-1} \times K\big(\alpha(u), \tu y, h(\alpha(u), \tu y), f(\alpha(u), \tu y)\big)
    $$
    with the convention that $K\big(\alpha(-1), \tu y, h(\alpha(u), \tu y), f(\alpha(u), \tu y)\big)= f(0, \tu y)=g(\tu y)$.
    Let $s_j(u)=\fun{sg}\big(f(\alpha(u), \tu y)\big)$, $j\in \{1, \dots, h\}$ be a list of signed terms appearing in $K$ and $c=\max_jc_j$.
    \\
    Let us consider the truncated sum between $a$ and $b-1$ bounded by $\ell(x)-1$.
    It is defined as:
    $$
    S(a,b) = \sum^{b-1}_{u=a} K\big(\alpha(u), \tu y, h(\alpha(u), \tu y), S(u)\big)
    $$
    with $S(a)=f\big(\alpha(a), \tu y\big)$.
    We consider the following $c+2$ sums in which an assumption is made on the value of the first term:
    $$
    S_i(a,b) = \sum^{b-1}_{u=a} 2^{\ell(x)-u-1}K_i\big(\alpha(u), \tu y, h(\alpha(u), \tu y), S_i(\alpha(u), \tu y)\big)
    $$
    for $i\in \{1,\dots, c+1\}$.
    Recall that $K_i$ is defined from $K$ by replacing each term $s_j(a)$ by its value assuming $S_0(a)\leq 0, S_i(u)=i$ for $i\in \{1,\dots c\}$ and $S_{c+1}(\alpha(a), \tu y)>c$.
    More precisely, for any $i$, $s_j(a)=1$ if $i>c_j$ and 0 otherwise.
    Then, for any $x$ and $\tu y$, the value of $f(x, \tu y)$ can be computed as follows:
    \begin{align*}
        f(x, \tu y) &= S(-1, \ell(x)) \\
        &= S\big(-1, \ell(x)/2\big) + S\big(\ell(x)/2, \ell(x)\big) \\
        &= f\big(\alpha(\ell(x)/2), \tu y\big) + S\big(\ell(x)/2, \ell(x)\big) \\
        &= \vee \begin{cases}
            0 + S_0\big(\ell(x)/2, \ell(x)\big) \\
            \dots \\
            c + S_c\big(\ell(x)/2, \ell(x)\big) \\
            S\big(-1, \ell(x)/2\big) + S_{c+1}\big(\ell(x)/2, \ell(x)\big).
        \end{cases}
    \end{align*}
    Hence,
    $$
    f(x, \tu y) = + \begin{cases}
        S\big(-1, \ell(x)/2\big) \\
        S_0\big(\ell(x)/2, \ell(x)\big) \times \fun{sg}\big(S\big(-1, \ell(x)/2\big)\big) \\
        S_1\big(\ell(x)/2, \ell(x)\big) \times \fun{sg}\big(S\big(-1, \ell(x)/2\big)-1\big) \times \fun{sg}\big(1-S\big(-1, \ell(x)/2\big)\big) \\
        \dots \\
        S_c \big(\ell(x)/2, \ell(x)\big) \times \fun{sg}\big(S\big(-1, \ell(x)/2\big)-c\big) \times \fun{sg}\big(c-S\big(-1, \ell(x)/2\big)\big) \\
        S_{c+1} \big(\ell(x)/2, \ell(x)\big) \times \fun{sg}\big(S\big(-1, \ell(x)/2\big)-c-1\big).
    \end{cases}
    $$
    Independently computing $S\big(-1, \ell(x)/2\big), S_0\big(\ell(x)/2, \ell(x)\big), \dots, S_{c+1}\big(\ell(x)/2, \ell(x)\big)$ can be done in parallel: as seen, we know the actual starting value of $S_0\big(\ell(x)/2, \ell(x)\big)$, which is $f\big(\alpha(u), \tu y\big)=0$ (so that this sum can be computed exactly as the starting sum $S\big(-1, \ell(x)/2\big)$ in at most $\ell(x)$ steps),
    while to compute each $S_i\big(\ell(x)/2, \ell(x)\big)=\sum^{\ell(x)-1}_{u=\ell(x)/2} 2^{\ell(x)-u-1} K_1(u)$ we make the assumption that the actual value of $f(\alpha(a), \tu y)$
\end{proof}

\end{toappendix}


We now move to the new characterization of $\NC^1$.
To establish it, we introduce an $\ell$-ODE schema defined by imposing $A=-1$ and $B\ge 0$, that is by generalizing the constraint defining $\schODE$ from $B=K\in \{0,1\}$ to $B$ taking any positive value.

\begin{defn}[Schema $\chODE$]
    Given $g:\Nat^p \to \Nat$ and $h:\Nat^{p+1} \to \Nat$, the function $f:\Nat^{p+1} \to \Nat$ is defined by $\chODE$ if it is the solution of the IVP with initial value $f(0, \tu y)=g(\tu y)$ and such that
    $$
        \frac{\partial f(x, \tu y)}{\partial \ell} = -f(x, \tu y) + B\big(x, \tu y, h(x, \tu y), f(x, \tu y)\big)
    $$
    where, for any $x, \tu y$, $B(x,\tu y, h(x, \tu y), f(x, \tu y)) \ge 0$ is a $\fun{sg}$-polynomial expression, and $f(x, \tu y)$ occurs in it only in expressions of the form $\fun{sg}(f(x, \tu y)-c)$, being $c$ a constant.
\end{defn}

\begin{lemma}\label{lemma:chODE}
    If $f(x, \tu y)$ is defined by \emph{$\ell$-bODE} from functions in $\NC^1$, then $f(x,\tu y)$ is in $\NC^1$ as well.
\end{lemma}
\begin{proof}
\begin{sloppypar}
    To simplify the exposition, w.l.o.g. let us suppose that: (i) $B$ includes only one signed term, say $s$, in which $f(x, \tu y)$ occurs, and (ii) $\ell(x)$ is a power of two (to avoid rounding).
    By definition of $\ell$-ODE, $f(x, \tu y)=B\big(\alpha(\ell(x)-1), \tu y, h(\alpha(\ell(x)-1), \tu y), f(\alpha(\ell(x)-1), \tu y)\big)= B_i(\ell(x)-1)$, with $i=s(\ell(x)-1)$.
    %
    %
    We consider partial or truncated results between $a$ and $b-1$ bounded by $\ell(x)-1$, intuitively corresponding to subsequent compositions of $s$ in $B$, going from $a$ to $b-1$; that is, either $C(a,b)=C_0(a,b)$, i.e.~the computation of $f(b-1, \tu y)$ starting from $f(a,\tu y)=B_0(a-1)$, or $C(a,b)=C_1(a,b)$, i.e.~the computation of $f(b-1, \tu y)$ starting from $f(a,\tu y)=B_1(a-1)$.
    The values of $C(0, a-1)$, $C_0(a,b)$ and $C_1(a,b)$ can be computed independently from all previous values of $f(\cdot,\tu y)$.
    Additionally, only one between $C_0(a,b)$ and $C_1(a,b)$ corresponds to the actual value of $C(a,b)$.
    Then, for any $x$ and $\tu y$, the value of $f(x, \tu y)$ can be computed inductively by selecting the correct partial $B(a,b)$ as follows:
    $f(x, \tu y) = C(-1, \ell(x)) 
    = \fun{sg}\big(C(-1, \ell(x)/2)\big) \times   C_1\big(\ell(x)/2, \ell(x)\big) +
    \fun{cosg}\big(C(-1, \ell(x)/2)\big) \times C_0\big(\ell(x)/2, \ell(x)\big)$.
    Both $C(-1, \ell(x)/2), C_0(\ell(x)/2, \ell(x))$ and $C_1(\ell(x)/2, \ell(x))$ as well as, for any $t\in \{0,\dots, \ell(x)-1\}$, $B_0(t),B_1(t)$ and the corresponding signs
    can be computed independently from each other and in parallel.
    The number of subsequent recursive calls needed to end up with the desired selection, i.e.~the one corresponding to $C_0(\ell(x)/2, \ell(x))$ and $C_1(\ell(x)/2, \ell(x))$, would require $\ell(x)/2$ (selection) steps, ensuring we are in $\NC^1$.
\end{sloppypar}
\end{proof}


\begin{remark}
    As desired, Lemma~\ref{lemma:chODE} still holds when dealing with $s_0, \dots, s_r$ signed expressions occurring in $B$.
    To deal with this more general case we need to add an initial consistency check, considering $\fun{sg}(B_{i_0,\dots, i_r}(t))$'s and $B_{i_0'\dots i_r'}(t+1)$'s, in the very beginning.
\end{remark}

Together with $\pODE$, $\chODE$ provides a new ODE-based characterization for $\NC^1$.

\begin{theorem}
    $\NC^1 = [\fun{0}, \fun{1}, \ell, \fun{sg}, +, -, \div 2, \#, \pi^p_i; \circ, \ell\emph{-pODE}, \ell\emph{-bODE}].$
\end{theorem}
\begin{proof}
    $(\supseteq)$ All basic functions of this algebra are already in $\AC^0$ and the closure property holds for $\circ, \pODE$ (generalization of Proposition~\ref{prop:pODE}) and $\chODE$ (Lemma~\ref{lemma:chODE}). 
    \\
    $(\subseteq)$ The proof is indirect and passes through Clote's algebra:
    $$
    \mathcal{N}_0' = [\fun{0}, \pi^p_i, \fun{s}_0, \fun{s}_1, \ell, \fun{BIT}, \#; \circ, \text{CRN}, \text{4-BRN}],
    $$
    where a function is defined by 4-BRN if $f(0, \tu y)=g(\tu y)$ and $f(s_i(x), \tu y)=h_i(x, \tu y, f(x, \tu y))$ for $g$ and $h$ taking values in $\{0,\dots, 4\}$ only.
    In~\cite{Clote93}, this class is proved to be the one characterizing $\NC^1$. 
    In particular, to show that any function defined by 4-BRN 
    can be rewritten as the solution of an instance of $\chODE$, we consider the IVP defined by the initial value $\fun{4brn}(x,y,\tu y)=g(\tu y)$ and such that:
    \footnotesize
    \begin{align*}
        \frac{\partial \fun{4brn}(x, y, \tu y)}{\partial \ell} = -\fun{4brn}(x, y, \tu y) &+ h_{z(x)}(\ell(x), \tu y, 0) \times \fun{cosg}\big(\fun{4brn}(x, y, \tu y)\big) \\
        & + h_{z(x)}(\ell(x), \tu y, 1) \times \fun{sg}\big(\fun{4brn}(x, y, \tu y)\big) \times \fun{cosg}\big(\fun{4brn}(x, y, \tu y) - 1\big) \\
        & + h_{z(x)}(\ell(x), \tu y, 2) \times \fun{sg}\big(\fun{4brn}(x, y, \tu y) - 1\big)\times \fun{cosg}\big(\fun{4brn}(x,y,\tu y)-2\big) \\
        &+ h_{z(x)}(\ell(x), \tu y, 3) \times   \fun{sg}\big(\fun{4brn}(x,y,\tu y) - 2\big) \times \fun{cosg}\big(\fun{4brn}(x, y, \tu y) - 3\big) \\
        &+ h_{z(x)}(\ell(x), \tu y, 4) \times \fun{sg}\big(\fun{4brn}(x,y,\tu y) - 3\big) \times \fun{cosg}\big(\fun{4brn}(x, y, \tu y) - 4\big)
    \end{align*}
    \normalsize
    \noindent
    with $z(x)=\fun{BIT}(\ell(y)-\ell(x)-1, y)$.
    Then, $\fun{4brn}(x,x,\tu z)=f(x,x,\tu z)$, where $f$ is defined by 4-BRN from $g$ and $h_i$.
    Since Clote's $\mathcal{N}_0'$ is nothing but $\mathcal{A}_0$ endowed with 4-BRN~\cite{Clote93} and by~\cite{ADK24a} $\mathcal{A}_0=\ACDL$, this encoding of 4-BRN via $\chODE$, together with Proposition~\ref{prop:aux}, concludes our proof.
\end{proof}

\subsection{The Class $\AC^1$}\label{sec:ACone}

When we consider a generalization of the non-strict schemas shown in $\TC^0$ from $B=K\in \{0,1\}$ to more liberal $B\ge 0$ (analogously to the generalization from $\schODE$ to $\chODE$), we obtain an $\ell$-ODE schema, which is provably in $\AC^1$.
Formally, it is defined by an IVP
with restrictions 
$A=0$ and, as mentioned, $B\ge 0$ with only simple calls to $f(x, \tu y)$.

The proof of the result below is similar  to that for Lemma~\ref{lemma:ODENC1}, except for the last step which amounts of a sum (possibly including carries) rather than a concatenation (see Appendix~\ref{appendix:AC}).

\begin{lemma}[Bounded Sum with Simple Calls]\label{lemma:AC1}
    Let $g: \Nat^p \to \Nat$ and $h:\Nat^{p+1}\to \Nat$ be computable in $\AC^1$. Then, computation through the schema defined by the IVP with initial value $f(0, \tu y)=g(\tu y)$ and such that
    $$
    \frac{\partial f(x, \tu y)}{\partial \ell} = B\big(x, \tu y, h(x, \tu y), f(x, \tu y)\big)
    $$
    where $B \ge 0$ is a $\fun{sg}$-polynomial expression with $f(x,\tu y)$ only appearing under the scope of the sign function inside expressions of the form $s=\fun{sg}(f(x, \tu y))$, is in $\AC^1$ as well.
\end{lemma}

\begin{toappendix}
\subsection{Schemas in $\AC^1$}\label{appendix:AC}
\begin{proof}[Proof of Lemma~\ref{lemma:AC1}]
First observe that, by definition of $\ell$-ODE, 
    $$
    f(x, \tu y) = \sum^{\ell(x)-1}_{u=-1} B\big(\alpha(u), \tu y, h(\alpha(u), \tu y), f(\alpha(u), \tu y)\big)
    $$
    with the convention that $B(\alpha(-1), \tu y, h(\alpha(-1), \tu y), f(\alpha(-1), \tu y))=g(\tu y)=f(0, \tu y)$.
    \\
    We can consider the truncated sums between $a$ and $b-1$ bounded by $\ell(x)-1$.
    It is defined by
    \begin{align*}
        S(a,b) = \sum^{b-1}_{u=a} B\big(\alpha(u), \tu y, h(\alpha(u), \tu y), S(\alpha(u), \tu y)\big) 
        = \sum^{b-1}_{u=a} B\big(\alpha(u), \tu y, h(\alpha(u), \tu y), s(\alpha(u))\big)
    \end{align*}
    with $S(\alpha(a), \tu y) = f(\alpha(a), \tu y)$ and $s(a)=\fun{sg}(f(\alpha(a-1), \tu y))$.
    By definition of $f(x, \tu y)$, if for some $t\in \{0, \dots, \ell(x)-1\}$, $f(t, \tu x)>0$, then, for all $t'\in \{t+1, \dots, \ell(x)-1\}$, $f(t, \tu x)>0$, that is $B(t')=B_1(t')$.
    Thus, we can consider the following two similar sums:
    \begin{align*}
        S_0(a,b) &= \sum^{b-1}_{u=a} B\big(\alpha(u), \tu y, h(\alpha(u), \tu y), S_0(\alpha(u), \tu y)\big) \\
        S_1(a,b) &= \sum^{b-1}_{u=a} B\big(\alpha(u), \tu y, h(\alpha(u), \tu y), S_0(\alpha(u), \tu y)\big) = \sum^{b-1}_{u=a} B_1\big(\alpha(u)\big)
    \end{align*}
    with $S_0(\alpha(a), \tu y)=0$ and $S_1(\alpha(a), \tu y)=1$.
    Clearly, $S(a,b)$ is either $S_0(a,b)$ or $S_1(a, b)$. 
    (Notably, in the former case $S(0, b)=S_0(a,b)$.)
    Then, for any $x, \tu y$, the value of $f(x, \tu y)$ can be computed inductively as follows:
    \begin{align*}
        f(x, \tu y) &= S\big(-1, \alpha(\ell(x))\big) \\
        &= S\big(-1, \alpha(\ell(x)/2)\big) + S\big(\alpha(\ell(x)/2), \alpha(\ell(x))\big) \\
        &= f\big(\alpha(\ell(x)/2), \tu y\big) + S\big(\alpha(\ell(x)/2), \alpha(\ell(x))\big) \\
        &= \begin{cases}
            0 + S_0\big(\alpha(\ell(x)/2), \alpha(\ell(x))\big) \\
            S\big(-1, \ell(x)/2\big) + \sum^{b-1}_{u=a} B_1(\alpha(u))
        \end{cases}
    \end{align*}
    Hence,
    $$
        f(x, \tu y) = + \begin{cases}
            S(-1, \alpha(\ell(x)/2)) \\
            S_0\big(\alpha(\ell(x)/2), \alpha(\ell(x))\big) \times 1- s\big(\alpha(\ell(x)/2)-1\big) \\
            S_1\big(\alpha(\ell(x)/2), \alpha(\ell(x))\big) \times s\big(\alpha(\ell(x)/2)-1\big)
        \end{cases}
    $$
    Computing $S\big(-1, \alpha(\ell(x)/2\big), S_0\big(\alpha(\ell(x)/2), \alpha(\ell(x))\big)$ and $S_1\big(\alpha(\ell(x)/2), \alpha(\ell(x))\big)$ can be done independently in parallel.
    In particular, as seen, we know that the actual starting value of $S_0\big(\alpha(\ell(x)/2), \alpha(\ell(x))\big)$ (which is $f(u, \tu y)=0$), while for $S_1\big(\alpha(\ell(x)/2), \alpha(\ell(x))\big)=\sum^{\ell(x)-1}_{u=\ell(x)/2}B_1(\alpha(u))$ we do not need to know the actual value of $f$ through the computation, since we know that it is always strictly greater than 0.
    Once $S(-1, \alpha(\ell(x)/2))=f(\alpha(\ell(x)/2), \tu y)$ is known, its sign can be tested and the outermost addition can be performed. 
    Testing the sign of $f(t, \tu y)$ can be done in $O(1)$.
    So, the depth of the process obeys the following recursive equation:
    $$
    depth(\ell(x)) = depth(\ell(x)/2) + O(1).
    $$
    It is known that this given $depth(\ell(x))=O(\ell_2(x))$, hence is logarithmic in the length $\ell(x)$ of one of the input.
\end{proof}
    
\end{toappendix}


\section{Characterizing Classes Deriving along $\ell_2$}\label{sec:2ell}

Here, we consider alternative linear ODEs in which derivation is not along $\ell$ but along the ``slower'' function $\loglog$. 
Of course, in this context, same syntactic constraints on the linearity of the schema end up characterizing smaller classes (for example, the solution of the system shown in $\AC^1$ when deriving along $\ell$ is shown to be computable in $\TC^0$ if  derivation is along $\ell_2$). 
The main result of this section is that linear ODEs with derivative along $\loglog$ 
are in $\NC^1$ (rather than $\FP$ as in the case of $\ell$~\cite{BournezDurand19}).

First, observe that when deriving along $\ell_2$, computation through a strict linear ODE schema is (of course in $\TC^0$ but even) in $\AC^0$ and corresponds to \textsc{LogItAdd}, i.e.~iterated addition of $\ell_2(n)$ numbers.

\begin{lemma}\label{lemma:logitadd}
    Let $g:\Nat^{p} \to \Nat$ and $h:\Nat^{p+1}\to \Nat$ be computable in $\AC^0$.
    Then, $f:\Nat^{p+1} \to \Nat$ defined from $h$ and $g$ by the IVP with initial value $f(0,\tu y)= g(\tu y)$ and such that
    $$
        \frac{\partial f(x, \tu y)}{\partial \ell_2} = \sB\big(x, \tu y, h(x, \tu y)\big)
    $$
    where $\sB$ is a limited $\fun{sg}$-polynomial expression,
    is in $\AC^0$.
\end{lemma}
\noindent
When this schema is generalized to the non-strict setting by allowing $f(x, \tu y)$ to occur in $B$ under the scope of the sign function, computation is proved to be in $\TC^0$ by a much evolved argument.

\begin{lemma}\label{lemma:l2TC}
        Let $g: \Nat^p \to \Nat$ and $h:\Nat^{p+1}\to \Nat$ be definable in $\TC^0$.
        Then, computation through the schema defined by the IVP with initial value 
            $f(0, \tu y) = g(\tu y)$
            and such that
            $$
            \frac{\partial f(x, \tu y)}{\partial \ell_2} = B\big(x, \tu y, h(x, \tu y), f(x, \tu y)\big)
            $$
            where $B$ is a $\fun{sg}$-polynomial expression 
        in which $f(x, \tu y)$ only appears under the scope of the sign function,
        is in $\TC^0$.
    \end{lemma}


\begin{toappendix}
%
%
\begin{lemma}\label{lemma:logitadd}
    Let $\sB$ be a limited $\fun{sg}$ polynomial expressions and $g:\Nat^{p} \to \Nat, h:\Nat^{p+1}\to \Nat$ be functions computable in $\AC^0$.
    Then, the function $f:\Nat^{p+1} \to \Nat$ defined from $\sB, h$ and $g$ by the IVP below:
    \begin{align*}
        f(0, \tu y) &= g(\tu y) \\
        \frac{\partial f(x, \tu y)}{\partial \ell_2} &= \sB\big(x, \tu y, h(x, \tu y)\big)
    \end{align*}
    is in $\AC^0$.
\end{lemma}
\begin{proof}
    By definition of $\ell_2$-ODE, the solution of the given system is of the form:
    $$
    f(x, \tu y) = \sum^{\ell_2(x)-1}_{u=-1} \sB\big(\alpha_2(u), \tu y, h(\alpha_2(u), \tu y)\big)
    $$
    where $\alpha_2(u)=2^{2^u-1}-1$ and with the convention that $\sB(\alpha_2(-1), \tu y, h(\alpha_2(-1), \tu y)=f(0, \tu y)= g(\tu y)$.
    Since $g, h$ (by hypothesis) and +, -, $\fun{sg}$ ($\sB$ is limited $\fun{sg}$-polynomial expression) are in $\AC^0$, each $\sB\big(\alpha_2(u), \tu y, h(\alpha_2(u))\big)$ is.
    There are at most $\ell_2(x)$ terms $\sB\big(\alpha_2(u), \tu y, h(\alpha_2(u))\big)$, so their sum, i.e.~$f(x, \tu y)$, can be computed in constant time (\textsc{LogItAdd} is $\AC^0$, see~\cite{Vollmer}). 
\end{proof}

    \begin{proof}[Proof of Lemma~\ref{lemma:l2TC}]
        To ease the presentation, we make the non-essential assumption that there is only one signed term $s(f(x, \tu y))=\fun{sg}(P(f(x, \tu y)))$ occurring in $B$, where $P$ is a polynomial expression.
        \\
        Clearly, each $s(f(x, \tu y))$ can take two values only, i.e.~0 or 1.
        Since derivation is along $\ell_2(x)$, the value of $f(x)$ depends on the \emph{distinct} values of $\ell_2$ between 0 and $x$, hence of each $\alpha_2(u)=2^{2^u-1}-1$, the greatest integer $t$ such that $\ell_2(t)=u$, for $u\in \{0,\dots, \ell_2(x)-1\}$.
        More precisely, 
        $$
            f(x, \tu y) = f(0, \tu y) + \sum^{\ell_2(x)-1}_{u=0} B(x, \tu y, h(x, \tu y), f(x, \tu y)).
        $$
        Notice that, for each $t\in \{0,\dots, \ell_2(x)\}$ it holds that:
        $$
             f(\alpha_2(t), \tu y) = f(0, \tu y) + \sum^{t-1}_{u=0} B\big(\alpha_2(u), \tu y, h(\alpha_2(u), \tu y), f(\alpha_2(u), \tu y)\big).
        $$
        Hence, each signed term is obtained through partial sums.
        To compute $f(x, \tu y)$, we will make assumptions on the values of the intermediate signed terms and check afterwards that the assumptions are consistent.
        For each value $u$, consider
        $$
            B_0(u) \text{ and } B_1(u)
        $$
        with $B_i(u)$, for $i\in \{0,1\}$ is $B(\alpha_2(u), \tu y, h(\alpha_2(u), \tu y), f(\alpha_2(u), \tu y))$ where $s(f(\alpha_2(u), \tu y))$ has been replaced by the corresponding $i\in \{0, 1\}$.
        \\
        Clearly, only one choice of sequence for the value of the signed term is consistent and the following holds:
        $$
        f(x, \tu y) = f(0, \tu y) + \sum_{\epsilon \in \{0,1\}^{\ell_2(x)}} \sum^{\ell_2(x)-1}_{u=0} B_{\epsilon_u}(u),
        $$
        where $\epsilon = (\epsilon_0, \dots, \epsilon_{\ell_2(x)-1})=(\epsilon_u)$.
        Let's now consider the partial sums:
        $$
        S_{\epsilon, t} = \sum^{t-1}_{u=0} B_{\epsilon_i}(u)
        $$
        for all $\epsilon \in \{0,1\}^{\ell_2(x)}$ and $t\in \{0,\dots, \ell_2(x)\}$.
        The correct choice of $\epsilon$ is guided by $s(f(0, \tu y))$, which must be equal to $\epsilon_0$ and by the fact that from one step to the other, the computed sign of the partial sum, i.e.~$s(S_{\epsilon,t})$ is equal to the next choice $\epsilon_t$.
        Then, to summarize, function $f(x, \tu y)$ satisfies the following equation (mixing Boolean and arithmetic operators: since one choice only is consistent, the external operator can be a disjunction):
        $$
            f(x, \tu y) = f(0, \tu y) + \bigvee_{\epsilon \in \{0,1\}^{\ell_2(x)}}\bigg(\bigwedge^{\ell_2(x)-1}_{t=0} (s(S_{\epsilon,t})=\epsilon_t)\bigg) \times S_{\epsilon, \ell_2(x)}.
        $$
        We now describe the construction and its complexity:
        \begin{itemize}
            \itemsep0em
            \item In the first layer, for any possible $t\in \{0, \dots, \ell_2(x)\}$ compute \emph{in parallel}, $B_i(t)$, that is:
            \begin{itemize}
                \itemsep0em
                \item $B_0(0)$ and $B_1(0)$
                \item $B_0(1)$ and $B_1(1)$
                \item $\dots$
                \item $B_0(\ell_1(x)-1)$ and $B_1(\ell_1(x)-1)$
            \end{itemize}
            In this case, since $B$ is a polynomial expression (i.e.~computable in $\TC^0$), each $B_i$ is obtained from $B$ by substitution of one of its input and the number of values to compute is $O(\ell_2(x))$, the implementation is of polynomial size and constant depth on a  circuit with \textsc{Maj} gates.
            \item For all $\epsilon \in \{0,1\}^{\ell_2(x)}$ and for all $t\in \{1,\dots, \ell_2(x)\}$, compute in parallel the sums:
            $$
                S_{\epsilon, t} = \sum^{t-1}_{u=0} B_{\epsilon_u}(u).
            $$
            Each $S_{\epsilon,t}$ is an iterated sum of at most $O(\ell_2(x))$ terms each computable by $\TC^0$ circuit.
            There are $O(\ell_2(x) \times 2^{\ell_2(x)})=O(\ell_2(x)\times \ell(x))$ such sums, hence a linear number in the input size.
            The whole step is in $\TC^0$
            \item For all $\epsilon \in \{0,1\}^{\ell_2(x)}$ and $t\in \{1,\dots, \ell_2(x)\}$ compute the sign $s(S_{\epsilon, t})$. 
            Since $s(S_{\epsilon,t})$ is a signed expression, i.e.~is a sign function applied to a polynomial expression and since addition and multiplication are in $\TC^0$, this step is also in $\TC^0$. 
        \end{itemize}
        The sequence of tests to evaluate $\bigwedge^{\ell_2(x)-1}_{t=0} (s(S_{\epsilon, t})=\epsilon_t)$ and then
        $$
            f(0, \tu y) + \bigvee_{\epsilon \in \{0,1\}^{\ell_2(x)}} \bigg(\bigwedge^{\ell_2(x)-1}_{t=0} (s(S_{\epsilon, t})=\epsilon_t)\bigg) \times S_{\epsilon, \ell_2(x)}
        $$
        can be wired in a constant depth and polynomial size circuit not even using majority gates.
        \\
        The whole computation is in $\TC^0$.
    \end{proof}
\end{toappendix}

As mentioned, the IVP defined deriving along $\ell_2$ and allowing ``full linearity'' is particularly interesting.
%

\begin{defn}[Schema $\ell_2$-ODE]
    Given $g:\Nat^p \to \Nat$ and $h:\Nat^{p+1} \to \Nat$, the function $f:\Nat^{p+1} \to \Nat$ is defined by $\ell_2$-ODE if it is the solution of the IVP with initial value $f(0, \tu y)=g(\tu y)$ and such that:
    $$
    \frac{\partial f(x, \tu y)}{\partial \ell_2} = A\big(x, \tu y, h(x, \tu y), f(x, \tu y)\big) \times f(x, \tu y) + B\big(x, \tu y, h(x, \tu y), f(x, \tu y)\big).
    $$
\end{defn}
\noindent
This schema  captures the computation of expressions in the following form:
\begin{align*}
    f(x, \tu y) &= \sum^{\ell_2(x)-1}_{u=-1} \Bigg(\prod^{\ell_2(x)-1}_{z=u+1} \bigg(1+A\Big(\alpha_2(z), \tu y, h(\alpha_2(z), \tu y), f(\alpha_2(z), \tu y)\Big)\bigg) \Bigg) \\
    &\quad \times B\Big(\alpha_2(u), \tu y, h(\alpha_2(u), \tu y), f(\alpha_2(u), \tu y)\Big),
\end{align*}
where $\alpha_2(v)=2^{2^v-1}-1$, i.e.~$\ell_2(\alpha_2(v))=v$ and $\alpha_2(v)$ is the greatest integer $t$ such that $\ell_2(t)=v$.

\begin{lemma}\label{lemma:lTwoODE}
    Let $f$ be defined by $\ell_2$-ODE from $g$ and $h$ in $\NC^1$, 
    then $f$ is in $\NC^1$. 
\end{lemma}

\begin{toappendix}
\bigskip 

To prove Lemma~\ref{lemma:lTwoODE}, we need to introduce some preliminary remarks and notational conventions.
The following remark will be largely exploited in the proof of Lemma~\ref{lemma:lTwoODE}.

\begin{remark}
    Since in the given system $A$ and $B$ are essentially constant in $f(\alpha_2(z), \tu y)$, each sub-term of $A$ and $B$ involving $f(\alpha_2(z), \tu y)$ is a $\fun{sg}$-polynomial expression of the form $\fun{sg}\big(P(\alpha_2(z), \tu y)\big)$.
    So the dependency is on Boolean value that themselves depend on $f(\alpha_2(z), \tu y)$.
\end{remark}
Again, to ease the presentation we make the following (non essential) assumptions:
\begin{itemize}
    \itemsep0em
    \item The system is of dimension 1, i.e.~over scalars
    \item In $A$ (resp. $B$) there is only one sub-expression involving $f(x, \tu y)$ denoted by
    $$
    s_1(f(x, \tu y)) = \fun{sg}\big(P_1(x, \tu y, f(x, \tu y))\big)
    $$
    resp., $s_2(f(x, \tu y))=\fun{sg}\big(P_2(x,\tu y, f(x, \tu y))\big)$.
\end{itemize}

\begin{notation}
    Given $i<\ell_2(x)$, let $t_i=\alpha_2(i)$, i.e.~$t_0=0$, $t_1=2^1-1=1$, $t_2=2^3-1 =7$ and, more generally, $t_i=2^{2^i-1}-1$.
    For $i\in \{0, \dots, \ell_2(x)-1\}$, let 
    \begin{align*}
        \tu a(i) &: s_1(f(t_i, \tu y)) \mapsto \{0,1\} \\
        \tu b(i) &: s_2(f(t_i, \tu y)) \mapsto \{0,1\}.
    \end{align*}
    Let
    $A_0(x, \tu y, h(x, \tu y))$ or even $A_0(x)$ (resp. $A_1(x, \tu y, h(x, \tu y))$ or $A_1(x)$) denote $A(x, \tu y, h(x, \tu y), f(x, \tu y))$ such that all occurrences of $s_1(f(x,\tu y))$ in it has been substituted by the value 0 (resp. 1).
    \\
    Let us define $f_{\tu a, \tu b}$ as follows:
    $$
        f_{\tu a, \tu b}(t_i) = \sum^{i-1}_{u=-1} \Bigg(\prod^{i-1}_{z=0} \bigg(1+A_{\tu a(z)}\Big(t_z, \tu y, h(t_z, \tu y)\Big)\bigg) \Bigg) \times B_{\tu b(u)} \big(t_u, \tu y, h(t_u, \tu y)\big)
    $$
    with the convention that $B_{\tu b(-1)} (\cdot, \cdot, \tu y) = f(0, \tu y)$ and $\prod^{x-1}_x \kappa(x)=1$.
    Intuitively, it holds that
    $$
    f_{\tu a, \tu b}(t_i) = f(t_i)
    $$
    for the choice of $\tu a, \tu b$ that passes the following consistency test: check that, for every $j<i$, 
    $s_1(f_{\tu a, \tu b}(t_j))=\tu a(j)$ and $s_2(f_{\tu a, \tu b}(t_j))=\tu b(j)$ for every $j<i$.
\end{notation}


We can now examine how to compute in $\NC^1$ a function $f(x, \tu y)$ defined by $\ell_2$-ODE from functions $g$ and $h$ in $\NC^1$.

\begin{proof}[Proof of Lemma~\ref{lemma:lTwoODE}]
    By definition of $\ell_2$-ODE,
    \begin{align*}
        f(x, \tu y) &= \sum^{\ell_2(x)-1}_{u=-1} \Bigg(\prod^{\ell_2(x)-1}_{z=u+1} \bigg(1+A\big(\alpha_2(z), \tu y, h(\alpha_2(z), \tu y), f(\alpha_2(z,\tu y),\tu y)\big)\bigg)\Bigg) \\
        & \quad \quad \times B\big(\alpha_2(u), \tu y, h(\alpha_2(u), \tu y), f(\alpha_2(u), \tu y)\big)
    \end{align*}
    where $\alpha_2(t)=2^{2^t-1}-1$.

    We show that this function can be computed in $\NC^1$ by constructing the corresponding log-depth circuit with bounded fan-in gates:
    \begin{itemize}
        \itemsep0em
        \item As a first layer, generate in parallel all possible $(\tu a,\tu b)\in \{0,1\}^{2\ell_2(x)}$ i.e (reordered) the vector.
        $$
            (\tu a(0), \tu b(0) \quad \tu a(1), \tu b(1) \quad \dots \quad \tu a(\ell_2(x)-1)), \tu b(\ell_2(x)-1)).
        $$
        There are $O(\ell^2(x))$ such choices.
        \item For each such $\tu a, \tu b$ computes also in parallel 
        $$
            f_{\tu a, \tu b}(t_1) \quad f_{\tu a, \tu b}(t_2) \quad \dots \quad f_{\tu a, \tu b}(t_i) \quad \dots \quad f_{\tu a, \tu b}(t_m).
        $$
        Each  $f_{\tu a, \tu b}(t_1) $ is an iterated sum of iterated products and can be computed in $\ell_2(x)$ depth. There are $O(\ell_2(x))$ such values to compute. 
        \item For $\tu a, \tu b$, and any $i\in \{1,\dots, \ell_2(x)-1\}$, compute
        $$
        s_1(f_{\tu a, \tu b}(t_i)\big) \quad s_2\big(f_{\tu a, \tu b}(t_i)\big).
        $$
        Each of these independent computations can be done with depth at most $\ell_2(x)$.
        \item Next, for any $i\in \{1,\dots, m-1\}$ local consistency is checked, i.e.~it is checked in parallel whether
        $$
        s_1(f_{\tu a, \tu b}(t_i)) = \tu a(i) \quad s_2(f_{\tu a, \tu b}(t_i)) = \tu b(i).
        $$
        This is done in constant depth.
        \item Now the circuit also checks global consistency,
        that is, for the given fixed choice of $(\tu a, \tu b)$, it is checked whether every local choice is consistent.
        %
        It is done in the form of a binary tree. Clearly, there is only one pair $\tu a, \tu b$ only one of these checks would be different from 0.
        \\
        This is done in log$_3(x)$.
        \item At this point the computation concludes returning the value of $f_{\tu a, \tu b}(x)$ corresponding to the (only) choice of $\tu a, \tu b$ leading to a positive global check. Roughly speaking this is done under the form of a binary tree which starts  comparing the result of the global check for two (consecutive) pairs $\tu a', \tu b'$ and $\tu a'',\tu b''$ and push further the correct value ($f_{\tu a', \tu b'}(x)$ or $f_{\tu a'', \tu b''}(x)$) if one of them has a positive global check or  a sequence of $0$ of appropriate length if none of them is. At the root of the tree, the correct value of $f(x)$ is outputted. 
\end{itemize}
\end{proof}

\end{toappendix}

\section{Conclusion}\label{sec:conclusion}

The objective of this paper was to start a new and uniform approach for studying circuit complexity through the lens of ODEs.
This approach allows to consider bounds on size or depth in circuit computation in terms of constraints on the linearity of the equations defining recursion schemas in two main contexts of derivatives: derivation along $\ell$ and along $\ell_2$.
Additionally, we have investigated completeness both improving and placing in a coherent setting the results obtained in~\cite{ADK24a} for $\AC^0$ and $\TC^0$, and providing what, to the best of our knowledge, are the first ODE-based characterizations for $\FACC$ and $\NC^1$.
%
%
%
In addition of being a simple and natural way to design algorithms,  ODE tools have also enabled us to obtain relatively simple and self-contained membership proofs and characterizations compared to the often convoluted constructions or external theories referenced in previous literature.

Many directions for future investigation remain open.
A natural one would be the study of characterizations for the entire $\AC$ and $\NC$ hierarchies. For this, one has to identify which ODE features or characteristic would help to jump from one depth level, say $(\log n)^i$, to the next. 
Furthermore, a possible next step of our research involves defining (strict) $\ell$-ODE schemas corresponding to complexity classes with counters, lying between $\mathbf{FACC[2]}$ and $\TC^0$ (i.e.~which correspond to computations executed by unbounded fan-in circuits featuring \textsc{Mod $p$} gates), ultimately allowing us to better understand their interrelationships and the concept of counting within the framework of small circuits.
Another relevant and intriguing direction is the exploration of the expressive power of schemas with derivative along $\ell_2$ in particular w.r.t. the problem of capturing complexity classes. Studying the relationship between $\ell$ and $\ell_2$ schemas might also be informative. Classical methods such as changes of variables sometimes help to go from one ODE schema to another but, on the other side, it seems that changing the function along which the derivative is taken gives a different point of view in terms of what can be expressed.
%
%
However, this study has just been initiated here, and a systematic investigation is left for future work.
%
Other challenging directions for subsequent research include generalizing our approach to the continuous setting by analyzing small circuit classes over the reals, following the path outlined in~\cite{BlancBournez23} for $\FP$ and $\mathbf{FPSPACE}$.
Developing proof-theoretical counterparts to our ODE-style algebras, in the form of natural rule systems inspired by the ODE design to syntactically characterize the corresponding classes, represents another promising avenue for exploration. 

\newpage
\bibliography{references}

$$
$$

\end{document}